\documentclass[journal,draftcls,onecolumn,12pt,twoside]{IEEEtran}

\usepackage{amssymb, amsthm, enumerate, array, graphicx, changepage}
\usepackage[utf8]{inputenc} 
\usepackage[T1]{fontenc}
\usepackage{url}
\usepackage{ifthen}
\usepackage{cite}
\usepackage[cmex10]{amsmath} 
\usepackage{tcolorbox}
\usepackage{multirow}

\newtheorem{theorem}{Theorem}
\newtheorem{proposition}{Proposition}
\newtheorem{corollary}{Corollary}
\newtheorem{lemma}{Lemma}
\newtheorem{definition}{Definition}
\newtheorem{example}{Example}

\newtheorem{remark}{Remark}
\newtheorem{conjecture}{Conjecture}

\newcommand{\CAC}{\mathrm{CAC}}
\newcommand{\CACe}{\mathrm{CAC^e}}
\newcommand{\MCCAC}{\text{MC-CAC}}
\newcommand{\HH}{\mathsf{H}}
\newcommand{\C}{\mathcal{C}}
\newcommand{\Z}{\mathbb{Z}}
\newcommand{\LL}{\mathsf{L}}

\title{Multichannel Conflict-Avoiding Codes for Expanded Scenarios}
\author{
{Tsai-Lien Wong, Kangkang Xu, Yuan-Hsun Lo,~\IEEEmembership{IEEE Member},\\ Kenneth W. Shum,~\IEEEmembership{IEEE Senior Member}, Yijin Zhang,~\IEEEmembership{IEEE Senior Member}}
\thanks{
This work was supported in part by the National Science and Technology Council of Taiwan under Grants 113-2115-M-110-003-MY2 and 114-2628-M-153-001-MY3, and in part by the National Natural Science Foundation of China under Grant 62071236.
\emph{(Corresponding author: Kangkang Xu)}
}
\thanks{T.-L. Wong is with the Department of Applied Mathematics, National Sun Yat-sen University, Taiwan. Email: tlwong@math.nsysu.edu.tw}
\thanks{K. Xu and Y.-H. Lo are with the Department of Applied Mathematics, National Pingtung University, Taiwan.  Email: ivykkxu107@gmail.com, yhlo0830@gmail.com}
\thanks{K. W. Shum is with the School of Science and Engineering, The Chinese University of Hong Kong at Shenzhen, Shenzhen 518100, China. Email: wkshum@cuhk.edu.cn}
\thanks{Y. Zhang is with the School of Electronic and Optical Engineering, Nanjing University of Science and Technology, Nanjing 210094, China. Email: yijin.zhang@gmail.com}
}

\begin{document}

\maketitle

\begin{abstract}
A conflict-avoiding code (CAC) of length $L$ and weight $w$ is used for deterministic multiple-access without feedback.
When the number of simultaneous active users is less than or equal to $w$, such a code is able to provide a hard guarantee that each active user has a successful transmission within every consecutive $L$ time slots.
Recently, CACs were extended to multichannel CACs (MC-CACs) over $M>1$ orthogonal channels with the aim of increasing the number of potential users that can be supported.
While most existing results on MC-CAC are derived under the assumption $M\geq w$, this paper focuses on the case $M<w$, which is more relevant to practical application scenarios.
In this paper, we first introduce the concept of exceptional codewords in MC-CACs.
By employing Kneser’s Theorem and some techniques from additive combinatorics, we derive a series of optimal MC-CACs.
Along the way, several previously known optimal CAC results are generalized.
Finally, our results extend naturally to AM-OPPTS MC-CACs and mixed-weight MC-CACs, two classes of relevant codes. 
\end{abstract}

\section{Introduction}\label{sec:intro}
\subsection{Background}
A conflict-avoiding code (CAC)~\cite{LT05} is a deterministic multiple access scheme in a single-channel system that operates without grant-based scheduling, time synchronization, or channel feedback.
Despite this simplicity, a CAC ensures that each active user has at least one conflict-free transmission within a specified time duration by leveraging its strong cross-correlation properties. 
So, compared with probabilistic schemes, CACs are better suited for low-cost distributed networks under sporadic traffic that pursue ultra-high deadline-constrained reliability~\cite{LDENK21,KVSTP25}.
Recently, with the aim of further increasing the maximal number of potential users that can be supported, CACs are extended to multichannel CACs (MC-CACs) in~\cite{CLW19,LSWZ21} by extending their cross-correlation properties in two-dimensional (2D) domains.
Note that protocol sequences~\cite{LWLZCX23} and rendezvous sequences~\cite{CZWZYL22} can also be regarded as variants of CACs designed to satisfy different performance guarantees, which involve different cross-correlation properties.



Consider a $K$-user multiple access system with $M$ orthogonal channels indexed by $\mathcal{I}_M\triangleq\{1,2,..., M\}$ and at most $w$ active users, where the time axis is partitioned into equal-length time slots, whose duration corresponds to the transmission time for one packet.
The relative time offset (in unit of a time slot) of user $k$, denoted by $\tau_k$, can be an arbitrary nonnegative integer.
Following~\cite{MM85}, we assumed that all the time offsets are unknown to all the users, and can never be learned as the users receive no feedback from the channel.
Tailored for this system, the adopted access scheme is defined as an MC-CAC (or a CAC when $M=1$), which is a collection of $K$ codewords, 
each of which consists of $w$ distinct elements in $\mathcal{I}_M\times \Z_L$
such that any two distinct codewords have at most one overlapping element for any relative shift in $\Z_L$ between them.
Here, $\Z_L\triangleq\{0,1,\ldots,L-1\}$ denotes the ring of residue modulo $L$.
Each user $k$ is assigned a unique codeword from this collection,
and sends out a packet over channel $m$ at time slot $t+\tau_k$ if and only if user $k$ is active and the assigned codeword contains an element $(m,t+\tau_k$ (mod $L$)).
As a result, 
each active user has at least one conflict-free transmission in a period of $L$ slots regardless of the integral relative time offsets.
Note that the design of MC-CACs can be easily extended to that for the fully asynchronous model (i.e., $\tau_k$ may be any positive real number for each user $k$)~\cite{MM85,ZLW16}.

\subsection{Preliminaries}


Consider a subset $S\in\mathcal{I}_M\times\Z_L$.
For $m\in\mathcal{I}_M$, define
\begin{align*}
S_m\triangleq\{t:\,(m,t)\in S\}.
\end{align*}
By viewing $S$ as a transmission pattern described in the previous subsection, $S_m$ represents the set of time slots during which packets are transmitted on the $m$-th channel.
Obviously, $|S|=\sum_{m\in\mathcal{I}_M}|S_m|$. 
For two indices $i,j\in\mathcal{I}_M$, define the $(i,j)$ \textit{set of differences} by
\begin{equation}\label{eq:array_difference}
D_S(i,j):=
\begin{cases}
    \{t_1-t_2\ (\text{mod }L):\, t_1\in S_i, t_2\in S_j\} & \text{if }i\neq j, \\
    \{t_1-t_2\ (\text{mod }L):\, t_1\in S_i, t_2\in S_j\}\setminus\{0\} & \text{if }i=j.
\end{cases}
\end{equation}
$D_S$ is called the \textit{array of differences} of $S$.
Two arrays of differences $D_S$ and $D_{S'}$ are said to be disjoint, denoted by $D_S \cap D_{S'}=\emptyset,$ if they are entry-wise disjoint.

The formal definition of MC-CACs is given below.

\begin{definition}\label{def:MCCAC}\rm
Let $M$, $L$ and $w$ be positive integers.
An MC-CAC $\C$ of weight $w$ in $\mathcal{I}_M\times\Z_L$ is a collection of $w$-subsets, also called codewords, of $\mathcal{I}_M\times \mathbb{Z}_L$ such that 
\begin{equation}\label{eq:MCCAC}
D_S \cap D_{S'}=\emptyset \quad \forall S, S'\in \C, S\neq S'.
\end{equation}
The equation~\eqref{eq:MCCAC} is called the \emph{disjoint-difference-array property}.
Let $\MCCAC(M,L,w)$ denote the class of all such MC-CACs.
\end{definition}

The maximum size of a code in $\MCCAC(M,L,w)$ is denoted by $K(M,L,w)$, i.e.,
\begin{align*}
K(M,L,w)\triangleq\max\{|\C|:\,\C\in\MCCAC(M,L,w)\}.
\end{align*}
A code $\C\in\MCCAC(M,L,w)$ is called \textit{optimal} if its code size reaches $K(M,L,w)$.

When $M=1$, an MC-CAC reduces to a (single-channel) CAC, where each codeword is simply a subset of $\Z_L$. 
Accordingly, MC-CAC$(1,L,w)$ and $K(1,L,w)$ reduce to CAC$(L,w)$ and $K(L,w)$, respectively.
In this case, \eqref{eq:MCCAC} can be simplified as
\begin{equation}\label{eq:CAC}
d^*(S)\cap d^*(S')=\emptyset \quad \forall S,S'\in\mathcal{C}, S\neq S',
\end{equation}
by defining $d^*(S)\triangleq\{a-b\, (\text{mod } L):\,a,b\in S, a\neq b\}$, which is called the set of \emph{(nonzero) differences} of $S$. 
The equation~\eqref{eq:CAC} is called the \emph{disjoint-difference-set property}.
Note that, in the identical-index case of~\eqref{eq:array_difference}, one has $D_S(i,i)=d^*(S_i)$.
A $w$-subset $S\subseteq\mathbb{Z}_L$ is said to be \textit{equi-difference} with \textit{generator} $g\in\Z_L^*$ if $S$ is of the form $\{0,g,2g,\ldots,(w-1)g\}$.
A CAC is called equi-difference if it entirely consists of equi-difference codewords.
Let $\CACe(L,w)\subset\CAC(L,w)$ denote the class of all equi-difference codes and $K^e(L,w)$ be the maximum size among $\CACe(L,w)$.
Obviously, $K^e(L,w)\leq K(L,w)$.




\subsection{Related Works}

A common design goal of MC-CACs (or CACs) is to maximize the code size, i.e., to find $K(M,L,w)$.
When $M=1$, it was shown in~\cite{SWC10} that $K(L,w)$ increases approximately with slope $(2w-2)^{-1}$ as a function of length $L$, for any fixed $w$.
Meanwhile, an asymptotically tight upper bound of $K(L,w)$ was given in~\cite{SW10}.
Based on some techniques in finite fields and some properties of cyclotomic cosets, \cite{MMSJ07,ZS22} provided some direct and recursive constructions of optimal CACs for general weights.
Recently, some of the above results were generalized in~\cite{LWXZ25} by the help of some techniques in additive combinatorics.
For small $w$, the exact value of $K(L,3)$ is completely determined by \cite{LT05,JMJTT07,MFU09,FLM10} for even $L$.
As for odd length, $K(L,3)$ is determined for $L$ being some particular prime~\cite{LT05} and some composite number with particular factors~\cite{Levenshtein07,FLS14,MZS14,MM17}.
If only equi-difference codewords are concerned, $K^e(L,w)$ is obtained for some particular $L$ with $w=3$ in \cite{WF13,LMSJ14} and with weight $w=4$ in \cite{LFL15,LMJ16}.
For more constructions and extensions of optimal CACs, one can refer to~\cite{Momihara07,HLS23,HLS24}.

When it comes to MC-CACs, upper bounds of $K(M,L,w)$ of weight $w=3,4$ and $M\geq w$ were derived in~\cite{LSWZ21}, in which some cases were proven to be optimal by direct constructions based on some specific combinatorial designs.
The exact value of $K(M=2,L,w=3)$ was determined in~\cite{LGZ21} for some particular length $L$.
Under the assumption that $M\geq w$ and the prime factors of $L$ are all larger than or equal to $2w-1$, \cite{WFLWG23} derived a general upper bound of $K(M,L,w)$ and obtained some optimal constructions.
Meanwhile, a special class of MC-CACs under the restriction that each user has at most one-packet per time slot (AM-OPPTS) was considered in~\cite{WFLWG23} as well.

\subsection{Contribution}

To date, all known optimal constructions of MC-CACs in the literature~\cite{LSWZ21,LGZ21,WFLWG23} rely on techniques from combinatorial design theory.
On the other hand, a variety of algebraic methods have been developed for constructing CACs (e.g., \cite{MMSJ07,SW10,LWXZ25}).
Motivated by this, in this paper we initiate the use of algebraic tools for constructing MC-CACs.
Meanwhile, most existing results on MC-CACs are derived under the assumption $M\geq w$, which may be undesirable in resource-limited systems. So this paper focuses on the case $M<w$.

The considered length of CACs and MC-CACs in this paper is of the form $L=ap_1^{r_1}\cdots p_n^{r_n}$, where $\gcd(a,p_i)=1$, $1\leq i\leq n$, for some $a$ and primes $p_1,\ldots,p_n$.
As $\gcd(a,p_i)=1$ for all $i$, we have $\Z_{L}\cong\Z_a\times\Z_{p_1^{r_1}}\times\cdots\times\Z_{p_n^{r_n}}$.
A natural bijection between $\Z_L$ and $\Z_a\times\Z_{p_1^{r_1}}\times\cdots\times\Z_{p_n^{r_n}}$ is via the Chinese Remainder Theorem (CRT)~\cite{IR90}, i.e., $\theta:\Z_{L}\to\Z_a\times\Z_{p_1^{r_1}}\times\cdots\times\Z_{p_n^{r_n}}$ by 
\begin{equation}\label{eq:CRT-correspondence}
\theta(x)=(x\,(\bmod\, a), x\,(\bmod\, p_1^{r_1}),\ldots,, x\,(\bmod\, p_n^{r_n})).
\end{equation}
Therefore, a $w$-subset in $\Z_L$ can be simply put as a $w$-subset in $\Z_a\times\Z_{p_1^{r_1}}\times\cdots\times\Z_{p_n^{r_n}}$. 
The remaining algebraic tools will be introduced in Sections~\ref{sec:CAC} -- \ref{sec:MC-CAC_M} as needed.

In the classical study of CACs, upper bounds on $K(L,w)$ are typically derived by characterizing the structure and number of exceptional codewords, leading to tighter bounds~\cite{SWC10,LWXZ25}.
Although this idea has also appeared in the context of MC-CACs~\cite{LSWZ21,LGZ21}, a systematic treatment is still lacking.
In this paper, we introduce a general definition of exceptional codewords for MC-CACs.
By applying Kneser’s theorem, we establish several structural properties and derive upper bounds on $K(M,L,w)$.
These results not only refine those in~\cite{WFLWG23}, which only considered the case $M\geq w$, but also extend the underlying proof techniques.

The technical results and main contributions of this paper are summarized below.
\begin{enumerate}
\item Define the so-called exceptional codewords for MC-CACs, and derive a series of structural properties by means of Kneser's theorem. 
These properties can be seen as generalizations of those of exceptional codewords for traditional CACs.
See Lemmas~\ref{lem:exceptional_new} and \ref{lem:exceptional_2D}. 
\item Derive a general upper bound on the maximum code size of an MC-CAC, $K(M,L,w)$, for $M<w$, where the code length $L$ is under the assumption that the prime factors of $L$ are all larger than or equal to $2w-1$.
This not only refines the previously-known results in the literature, but also extend the underlying proof techniques.
\item Obtain two classes of optimal CACs of lengths $L=p_1^{r_1}\cdots p_n^{r_n}$ and $(w-1)p_1^{r_1}\cdots p_n^{r_n}$ and weight $w$, where $p_i$ are primes.
Based on these two results, we further obtain optimal MC-CACs of length $L=(w-1)p_1^{r_1}\cdots p_n^{r_n}$ and weight $w$ for $M=2$ channels. 
We also derive the lower and upper bounds of $K(M,L,w)$ when $M$ divides $w$, where the code length is $L=(2\frac{w}{M}-1)p_1^{r_1}\cdots p_n^{r_n}$.
\item Extend our results to AM-OPPTS MC-CACs. 
The construction of these codes are extended to mixed-weight MC-CACs, which are defined for the first time in this paper.
\end{enumerate}

The rest of this paper is organized as follows.
We define exceptional codewords for MC-CACs and provide some structural properties about them in Section~\ref{sec:Kneser}.
A general upper bound of $K(M,L,w)$ is derived in Section~\ref{sec:upper-bound}.
Based on two new class of optimal CACs given in Section~\ref{sec:CAC}, we obtain a class of MC-CACs for two channels in Section~\ref{sec:MC-CAC_2}.
Section~\ref{sec:MC-CAC_M} investigates MC-CACs for $M$ channels when $M$ divides $w$.
AM-OPPTS MC-CACs and mixed-weight MC-CACs are discussed in Section~\ref{sec:extension}.
Some concluding remarks including the growth rate of $K(M,L,w)$ are given in Section~\ref{sec:conclusion}.


\section{Exceptional Codewords}\label{sec:Kneser}

In a CAC of weight $w$, a codeword $S$ is called \textit{exceptional} if $|d^*(S)|<2w-2$~\cite{SW10,SWC10,LWXZ25}.
In this section, we will extend the concept of exceptional codewords to MC-CACs, together with some characterizations of structural properties.

\subsection{Exceptional Pairs and Revisit of the Kenser's Theorem}

We first give some notion of Additive Combinatorics~\cite{TV06}.
For two subsets $A,B\subseteq\mathbb{Z}_L$ and an element $x\in\mathbb{Z}_L$, define
\begin{align}
x+A &\triangleq \{x+a:\,a\in A\}, \notag \\
A+B &\triangleq \{a+b:\,a\in A, b\in B\}, \text{ and} \label{eq:set-addition} \\
A-B &\triangleq \{a-b:\,a\in A, b\in B\}. \notag
\end{align}
Moreover, define
\begin{align*}
d(A)\triangleq A-A.
\end{align*}
Note that $0\in d(A)$ and $d(A)\setminus\{0\}=d^*(A)$, the set of nonzero differences of $A$.
In this notion, the definition of exceptional codewords $S$ in a CAC can be represented by $|S-S|\leq 2|S|-2$.
We first extend this concept to a pair of subsets.

\begin{definition}\label{def:exceptional}\rm
A pair of two subsets $\{A,B\}$ is called \textit{exceptional} if 
\begin{equation}\label{eq:exceptional-two}
|A-B|\leq |A|+|B|-2.
\end{equation}
Note that $|A-B|=|B-A|$, so there is no confusion in writing $A-B$ or $B-A$ in~\eqref{eq:exceptional-two}.
\end{definition}

Let $T$ be a non-empty subset in $\mathbb{Z}_L$.
The set of \textit{stabilizers} of $T$ in $\mathbb{Z}_L$ is defined as
\begin{align*}
\HH(T) \triangleq \{h\in\mathbb{Z}_L:\,h+T=T\}.
\end{align*}
It is obvious that $0\in\HH(T)$ and $\HH(T)$ is a subgroup of $\mathbb{Z}_L$.
So, $|\HH(T)|$ divides $L$ by Lagrange's theorem.
$T$ is called \textit{periodic} if $\HH(T)$ is non-trivial, namely, $\HH(T)\neq\{0\}$.
Here we list some well-known properties (e.g., see~\cite{SWC10,LWXZ25}) of the set of stabilizers. 

\begin{proposition}\label{prop:stabilizer}
Let $T\subseteq\mathbb{Z}_L$ be non-empty.
\begin{enumerate}[(i)]
\item $\HH(T)$ is a subgroup of $\mathbb{Z}_L$, and thus $|\HH(T)|$ divides $L$.
\item If $0\in T$, then $\HH(T)\subseteq T$.
\item If $T$ is periodic, then $T=\bigcup_{a\in T}(a+\HH(T))$. Moreover, $|\HH(T)|$ divides $|T|$.
\end{enumerate}
\end{proposition}

The following lists some results about the number of stabilizers of the difference set of an exceptional subset, which can be seen in~\cite[Corollary 1 and Lemma 1]{LWXZ25}.

\begin{lemma}[\!\!\cite{LWXZ25}]\label{lem:exceptional_old}
Let $A$ be a nonempty subset in $\Z_L$.
If $A$ is exceptional, then
\begin{enumerate}[(i)]
\item $2\leq |\HH(A-A)|\leq 2|A|-2$,
\item $|\HH(A-A)|$ does not divide $2|A|-1$, and
\item $|\HH(A-A)|$ does not divide $|A|-1$.
\end{enumerate}
\end{lemma}

The above results were based on Kneser's theorem, in which the formal description is given below.

\begin{theorem}[\!\!\cite{Kneser53,TV06}]\label{thm:Kneser}
Let $A$ and $B$ be two non-empty subsets in $\Z_L$, and let $\HH=\HH(A+B)$.
Then,
\begin{equation}\label{eq:Kneser1}
|A+ B| \ge |A+\HH| +|B+\HH|-|\HH|.
\end{equation}
In particular,
\begin{equation}\label{eq:Kneser2}
|A+B|\ge |A|+|B|-|\HH|.
\end{equation}
\end{theorem}

We shall study the number of stabilizers of the difference set of an exceptional pair, which can be seen as generalizations of Lemma~\ref{lem:exceptional_old}.

\begin{lemma}\label{lem:exceptional_new}
Let $A,B$ be two nonempty subsets of $\mathbb{Z}_L$.
If $\{A,B\}$ is an exceptional pair, then 
\begin{enumerate}[(i)]
\item $2\leq|\HH(A-B)|\leq|A|+|B|-2$, and
\item $|\HH(A-B)|$ does not divide $|A|+|B|-1$.
\item $|\HH(A-B)|$ does not divide $|A|-1$ when $|A|=|B|$.
\end{enumerate}
\end{lemma}
\begin{proof}
For notational convenience, denote by $\HH=\HH(A-B)$.

(i) By~\eqref{eq:exceptional-two} and \eqref{eq:Kneser2}, we have
\begin{align*}
|A|+|B|-2 &\geq |A-B|=|A+(-B)|\\
&\geq |A|+|-B|- |\HH| = |A|+|B|- |\HH|,
\end{align*}
which implies that $|\HH|\geq 2$.
Therefore, the set $A-B$ is periodic.
By Proposition~\ref{prop:stabilizer}(iii), we get $|\HH|\leq |A-B|$, which is less than or equal to $|A|+|B|-2$ because $\{A,B\}$ is exceptional.

(ii) Suppose to the contrary that $|\HH|$ divides $|A|+|B|-1$. 
Since, by Proposition~\ref{prop:stabilizer}(i), $\HH$ is a subgroup of $\Z_L$, we have $\HH=-\HH$, which implies that $|-B+\HH|=|-(B+\HH)|=|B+\HH|$.
Moreover, as $B+\HH$ is the disjoint union of cosets of $\HH$, the size $|B+\HH|$ is a multiple of $|\HH|$, and thus $|B+\HH|\geq|\HH|\cdot\left\lceil|B|/|\HH|\right\rceil$.
Similarly, $|A+\HH|\geq|\HH|\cdot\left\lceil|A|/|\HH|\right\rceil$.
By plugging $A=A$, $B=-B$ into~\eqref{eq:Kneser1} yields that 
\begin{align*}
|A-B|\ge|A+\HH|+|B+\HH|-|\HH|\geq |\HH|\left(\left\lceil\frac{|A|}{|\HH|}\right\rceil + \left\lceil\frac{|B|}{|\HH|}\right\rceil - 1 \right).
\end{align*}
Since $\{A,B\}$ is an exceptional pair, namely, $|A-B| \leq | A|+|B|-2$, it follows that
\begin{equation}\label{eq:exceptional_new_proof1}
|A|+|B|-2 \geq |\HH|\left(\left\lceil\frac{|A|}{|\HH|}\right\rceil + \left\lceil\frac{|B|}{|\HH|}\right\rceil - 1 \right).
\end{equation}

Let $|A|=q_1|\HH| + r_1$ and $|B|=q_2|\HH| + r_2$, for some integers $q_1,q_2\geq 0$ and $0<r_1,r_2\leq |\HH|$. 
Note that 
\begin{equation}\label{eq:exceptional_new_proof2}
0<r_1+r_2\leq 2|\HH|.
\end{equation}
By~\eqref{eq:exceptional_new_proof1}, we have
\begin{align*}
(q_1+q_2)|\HH| + r_1+r_2 - 2 = |A|+|B|-2 \geq |\HH|\big((q_1+1) + (q_2+1) - 1 \big),
\end{align*}
which implies that
\begin{equation}\label{eq:exceptional_new_proof3}
r_1  + r_2 -2 \geq |\HH|.
\end{equation}
Combining \eqref{eq:exceptional_new_proof2} and \eqref{eq:exceptional_new_proof3} yields
\begin{equation}\label{eq:exceptional_new_proof4}
|\HH|+1\leq r_1+r_2-1\leq 2|\HH|-1.
\end{equation}

By the assumption that $|\HH|$ divides $|A|+|B|-1=|\HH|(q_1+q_2)+r_1+r_2-1$, we get $|\HH|$ divides $r_1+r_2-1$.
This contradicts to \eqref{eq:exceptional_new_proof4}.

(iii) This case can be dealt with in the same argument we provided for (ii).
By letting $|A|=|B|=q|\HH|+r$, i.e., $q=q_1=q_2$ and $r=r_1=r_2$, the two inequalities \eqref{eq:exceptional_new_proof2} and \eqref{eq:exceptional_new_proof3} imply
\begin{equation}\label{eq:exceptional_new_proof5}
\frac{|\HH|}{2} \leq r-1 \leq |\HH|-1.
\end{equation}
Suppose to the contrary that $|\HH|$ divides $|A|-1=q|\HH|+r-1$.
We get $|\HH|$ divides $r-1$, which is a contradiction to \eqref{eq:exceptional_new_proof5}.
Hence the proof is completed.
\end{proof}

\subsection{Exceptional Codewords in MC-CACs}

For a subset $S\in\mathcal{I}_M\times\Z_L$, recall that $S_m$ collects time slots during which packets are transmitted on the $m$-th channel.
We further define
\begin{align*}
e_S\triangleq|\{m\in\mathcal{I}_M:\,S_m\neq\emptyset\}|
\end{align*}
indicating the number of channels that have at least one packet in the transmission pattern $S$.
With the notation given in~\eqref{eq:set-addition}, the $(i,j)$ set of difference of $S$, $D_S(i,j)$, defined in~\eqref{eq:array_difference} can be rewritten as $S_i-S_j$ if $i\neq j$.

We are ready to define exceptional subsets in $\mathcal{I}_M\times\Z_L$.


\begin{definition}\label{def:exceptional_multichannel}\rm
A subset $S\in\mathcal{I}_M \times \mathbb{Z}_L$ is called \textit{exceptional} if at least one of the followings holds.
\begin{enumerate}[(i)]
\item There exists $i\in\mathcal{I}_M$ such that $S_i$ is exceptional, i.e., $|S_i-S_i|\leq 2|S_i|-2$.
\item There exist distinct indices $i,j\in\mathcal{I}_M$ such that $\{S_i,S_j\}$ is exceptional, i.e., $|S_i-S_j|\leq |S_i|+|S_j|-2$.
\end{enumerate}
\end{definition}

\begin{lemma}\label{lem:exceptional_2D}
Let $S$ be a $w$-subset in $\mathcal{I}_M\times \mathbb{Z}_L$. 
If the prime factors of $L$ are all larger than or equal to $2w-1$, then $S$ is not exceptional.
\end{lemma}
\begin{proof}
Suppose to the contrary that $S$ is exceptional. 
Let $T$ be any nonempty subset in $\Z_L$. 
Since all prime factors $L$ are all larger than or equal to $2w-1$, by Proposition~\ref{prop:stabilizer}(i), we have
\begin{equation}\label{eq:exceptional_2D-proof}
|\HH(T)|=1\quad \text{or} \quad |\HH(T)|\geq 2w-1.
\end{equation}

If there exists an index $i\in\mathcal{I}_M$ such that $S_i$ is exceptional, by Lemma~\ref{lem:exceptional_old}(i), we have $2\leq |\HH(S_i-S_i)|\leq 2|S_i|-2\leq 2w-2$.
If there exist distinct indices $i,j\in\mathcal{I}_M$ such that $\{S_i,S_j\}$ is exceptional, by Lemma~\ref{lem:exceptional_new}(i), we have $2\leq |\HH(S_i-S_j)|\leq |S_i|+|S_j|-2 \leq w-2$.
Either case contradicts to~\eqref{eq:exceptional_2D-proof}, so the proof is done.
\end{proof}

We end this section with a fundamental result in Group Theory, which will be used in subsequent sections.

\begin{proposition}\label{prop:unique_subgroup}\rm
The subgroup of $\mathbb{Z}_L$ is uniquely determined by its order.
More precisely, for any divisor $d$ of $L$, the unique subgroup of $\mathbb{Z}_L$ with order $d$ is $\{0,L/d,2L/d,\ldots,(d-1)L/d\}$.
\end{proposition}

\section{General Upper Bounds on the Maximum Code Size}\label{sec:upper-bound}

We first recall an upper bound on the maximum code size of a CAC with some restriction on the code length.

\begin{theorem}[\!\cite{SW10}, Theorem 5]\label{thm:UB-single}
Let $L$ and $w$ be positive integers.
If the prime factors of $L$ are all larger than or equal to $2w-1$, then
\begin{align*}
K(L,w)\leq\frac{L-1}{2w-2}.
\end{align*}
\end{theorem}


The following result is due to \cite{WFLWG23}.

\begin{theorem}[\!\!\cite{WFLWG23}, Corollary 3.6]\label{thm:UB-any-w}
Let $M,L$ and $w$ be positive integers.
If the prime factors of $L$ are all larger than or equal to $2w-1$, we have
\begin{equation}\label{eq:UB-any-w}
K(M,L,w)\leq \frac{M(M-1)L}{w(w-1)}+\frac{M(L-1)}{2w-2}.
\end{equation}
\end{theorem}

The upper bounds in Theorem~\ref{thm:UB-any-w} can be improved in the case when $M<w$, which is shown in the following theorem.
For convenience, let $\Z_L^*=\Z_L\setminus\{0\}$ be the set of nonzeros in $\Z_L$.

\begin{theorem}\label{thm:UB-w>M}
Let $M,L$ and $w$ be positive integers with $M<w$. 
If the prime factors of $L$ are all larger than or equal to $2w-1$, we have
\begin{equation}\label{eq:K-UB-w>M}
K(M,L,w)\leq \frac{M(M-1)L}{(2w-M)(w-1)}+\frac{M(L-1)}{2w-2}.
\end{equation}
\end{theorem}
\begin{proof}
Let $\C$ be any code in $\MCCAC(M,L,w)$.
The codewords in $\C$ are partitioned, according to the number of channels that are occupied in their transmission patterns, i.e., $e_S$, into the following three classes:
\begin{align*}
   &\C_1=\{S\in \C: e_S=1\},\\
   &\C_2=\{S\in \C:e_S=M\}, \text{ and}\\
   &\C_3=\{S\in \C: 2\leq e_S\leq M-1\}.\\
\end{align*}

Since the prime factors of $L$ are all larger than or equal to $2w-1$, by Lemma~\ref{lem:exceptional_2D}, each codeword in $\C$ is non-exceptional.
Then, we have $|S_i-S_j|\geq |S_i|+|S_j|-1$ for any $i,j\in\mathcal{I}_M$, and in particular, $|d^*(S_i)|=|S_i-S_i|-1\geq 2(|S_i|-1)$.

For any $i\in\mathcal{I}_M$, by the disjoint-difference-array property, we have
\begin{align*}
    L-1 = |\Z^*_L| \geq \sum_{S\in\C}|D_S(i,i)|.
\end{align*}
Hence, 
\begin{align*}
    M(L-1) &\geq \sum_{i\in\mathcal{I}_M}\sum_{S\in \C}|D_S(i,i)| \\
    &= \sum_{S\in \C}\sum_{i\in\mathcal{I}_M, S_i\neq\emptyset}|d^*(S_i)| \geq \sum_{S\in \C}\sum_{i\in\mathcal{I}_M, S_i\neq\emptyset} 2(|S_i|-1)\\
    &= 2(w-1)|\C_1|+2(w-M)|\C_2|+2\sum_{S\in \C_3}(w-e_S),
\end{align*}
which implies that
\begin{equation} \label{eq:UB-w>M-1}
(w-1)|\C_1|+(w-M)|\C_2|+\sum_{S\in \C_3}(w-e_S)\leq \frac{M(L-1)}{2}.
\end{equation}

Now we consider $D_S(i,j)$ for distinct indices $i,j\in\mathcal{I}_M$, by going through all codewords in $\C$.
The sets of differences $D_S(i,j)$ for $S\in\C$ are mutually disjoint, and their union is a subset of $\Z_L$.
As $D_S(i,j)=\emptyset$ for $S\in\C_1$, we have
\begin{equation}\label{eq:UB-w>M-Ds}
    L = |\Z_L| \geq \sum_{S\in\C_2\cup\C_3}|D_S(i,j)|.
\end{equation}
Therefore,
\begin{align*}
    M(M-1)L&\geq \sum_{i,j\in\mathcal{I}_M, i\neq j}\sum _{S\in \C_2\cup \C_3} |D_S(i, j)| \\
    &\geq \sum_{S\in\C_2\cup\C_3}\sum_{\substack{i,j\in\mathcal{I}_M, i\neq j\\ S_i,S_j\neq\emptyset}}(|S_i|+|S_j|-1) \\
    &= \sum_{S\in\C_2\cup\C_3} \left(2(e_S-1)\left(\sum_{i\in\mathcal{I}_M, S_i\neq\emptyset}|S_i|\right)-e_S(e_S-1)\right) \\
    &= \sum_{S\in\C_2\cup\C_3}(e_S-1)(2w-e_S) \\
    &\geq (M-1)(2w-M)|\C_2| + \sum_{S\in\C_3}(e_S-1)(2w-e_S),
\end{align*}
yields that 
\begin{equation}\label{eq:UB-w>M-2}
    (M-1)(2w-M)|\C_2| + \sum_{S\in\C_3}(e_S-1)(2w-e_S) \leq M(M-1)L .
\end{equation}

We multiply the inequality in~\eqref{eq:UB-w>M-1} by $(2w-M)$ and combine the result with~\eqref{eq:UB-w>M-2}.
We then get 
\begin{equation}\label{eq:UB-w>M-3}
\begin{split}
& (w-1)(2w-M)\left(|\C_1|+|\C_2|\right) + \sum_{S\in\C_3}\left(-e_S^2+(M+1)e_S+2w^2-Mw-2w\right)  \\
\leq &\ \frac{M(L-1)}{2}(2w-M)+M(M-1)L. 
\end{split}
\end{equation}
When $M=2$, we have $\C_3=\emptyset$, and thus it follows from~\eqref{eq:UB-w>M-3} that
\begin{align*}
|\C|=|\C_1|+|\C_2| \leq \frac{M(M-1)L}{(2w-M)(w-1)}+\frac{M(L-1)}{2w-2},
\end{align*}
as desired.
So, we assume $M\geq 3$ in what follows.

Define a function $f(x)=-x^2+(M+1)x+2w^2-Mw-2w$.
The derivative of $f(x)$ is $-2x+M+1$, implying that $f(x)$ is a concave function with a unique maximum at $x=\frac{M+1}{2}$.
Let us focus on the interval $[2,M-1]$.
Since $\frac{M+1}{2}\in[2,M-1]$ for $M\geq 3$ and $f(2)=(2w-M)(w-1)+M-2=f(M-1)$, we have
\begin{align*}
f(x) &\geq \min\left\{f(2),f(M-1)\right\}\\ 
&= (2w-M)(w-1)+M-2,
\end{align*}
for $2\leq x\leq M-1$.
Therefore, 
\begin{equation}\label{eq:UB-w>M-4}
    f(e_S)\geq (2w-M)(w-1)+M-2 \geq (2w-M)(w-1)
\end{equation}
since $2\leq e_S\leq M-1$ and $M\geq 3$ for $S\in\C_3$.
Therefore, the left-hand-side of~\eqref{eq:UB-w>M-3} turns to
\begin{align}
& \quad \ (w-1)(2w-M)\left(|\C_1|+|\C_2|\right)+\sum_{S\in\C_3}f(e_S) \notag \\
&\geq (w-1)(2w-M)\left(|\C_1|+|\C_2|+|\C_3|\right) \label{eq:UB-w>M-5} \\
&= (w-1)(2w-M)|\C|. \notag
\end{align}
Then the inequality in~\eqref{eq:K-UB-w>M} follows.
\end{proof}

\begin{remark}\rm
If $M\geq 3$, the equality~\eqref{eq:K-UB-w>M} will hold only when $\C_3=\emptyset$, that is, all involved codewords $S$ has either $e_S=1$ or $e_S=M$.
This fact can be seen from \eqref{eq:UB-w>M-4} -- \eqref{eq:UB-w>M-5}: If $M\geq 3$, the inequality in~\eqref{eq:UB-w>M-4} turns out to be $f(e_S)>(2w-M)(w-1)$, which leads to a strict inequality in~\eqref{eq:UB-w>M-5} whenever $\C_3\neq\emptyset$.
\end{remark}

The upper bound stated in Theorem~\ref{thm:UB-any-w} and Theorem~\ref{thm:UB-w>M} remain valid when $M=1$.  So, both they can be viewed as generalizations of Theorem~\ref{thm:UB-single}.

It is worth mentioning that the proof method used for Theorem~\ref{thm:UB-w>M} is also applicable to Theorem~\ref{thm:UB-any-w} but it differs from the original proof given in~\cite{WFLWG23}. 
For completeness, we provide the new proof of Theorem~\ref{thm:UB-any-w} in Appendix~\ref{Appendix:A}.
Moreover, our proof in Appendix~\ref{Appendix:A} can further reveal the fact that the equality in~\eqref{eq:UB-any-w} will hold only when the codewords $S$ with either $e_S=1$ or $e_S=w$ are considered.
Note that all the optimal MC-CACs obtained in~\cite{LSWZ21,WFLWG23} have this property.

\section{Single-Channel CACs}\label{sec:CAC}

Firstly, we introduce the $p$-ary representation of a positive integer.
Given a positive integer $n$, let 
\begin{align*}
\mathbb{Z}^{\times}_n\triangleq\{x\in\mathbb{Z}_n:\,\gcd(x,n)=1\}.
\end{align*}
$\mathbb{Z}^{\times}_n$ is the set of \emph{units} (i.e., invertible elements) in $\mathbb{Z}^*_n$, and thus is a multiplicative group.
Note that $\mathbb{Z}^{\times}_n=\mathbb{Z}^{*}_n$ when $n$ is a prime.

Let $p$ be an odd prime and $r$ a positive integer.
For $c\in\mathbb{Z}_{p^r}$, consider the $p$-ary representation $c=c_0+c_1p+\cdots+c_{r-1}p^{r-1}$.
For $t=0,1,\ldots,r-1$, let $\mathsf{L}_t$ be the collection of $c\in\mathbb{Z}^{*}_{p^r}$ whose nonzero least significant digit in its $p$-ary representation is $p^t$.
Obviously, $|\mathsf{L}_t|=(p-1)p^{r-t-1}$, and $\mathsf{L_0},\mathsf{L}_1,\ldots,\mathsf{L}_{r-1}$ form a partition of $\mathbb{Z}^{*}_{p^r}$, i.e., $\mathbb{Z}_{p^r}^{*}=\mathsf{L}_0\uplus\mathsf{L}_1\uplus\cdots\uplus\mathsf{L_{r-1}}$.
Integers in $\mathsf{L}_t$ are called in the $t$-th \textit{layer}.

For a non-empty $A\subseteq\mathbb{Z}^{*}_p$, we arise it to a subset in $\mathbb{Z}^{*}_{p^r}$, for any positive integer $r$, by defining
\begin{equation}\label{eq:S_r(A)-def}
\mathcal{S}_r(A) \triangleq A_0\uplus A_1\uplus \cdots \uplus A_{r-1},
\end{equation}
where 
\begin{equation*}
A_t = \{c\in\mathsf{L}_t:\,c_t\in A\}.
\end{equation*}
$\mathcal{S}_r(A)$ is the collection of elements in $\mathbb{Z}^{*}_{p^r}$ whose nonzero least significant digit values in their $p$-ary representation are in $A$.
Obviously, $|A_t|=|A|p^{r-1-t}$ for each $t$, and thus 
\begin{equation}\label{eq:S_r(A)-size}
|\mathcal{S}_r(A)|=|A|\left(1+p+\cdots+p^{r-1}\right) =|A|\frac{p^r-1}{p-1}.
\end{equation}

The following is a useful property of the $p$-ary representation of $c\in\Z_{p^r}^*$, which can also be found in~\cite[Proposition 3]{LWXZ25}.

\begin{proposition}[\!\!\cite{LWXZ25}]\label{prop:p-ary}
Let $p$ be an odd prime and $r$ be a positive integer.
For $j\in\mathsf{L}_0$ and $c\in\mathsf{L}_t$, $0\leq t\leq r-1$, one has $jc\in\mathsf{L}_t$ and
\begin{equation}\label{eq:p-ary}
(jc)_t = j_0\cdot c_t \ (\bmod\, p).
\end{equation}
Note that $\mathsf{L}_0=\mathbb{Z}^{\times}_{p^r}$, which is the set of units in $\mathbb{Z}^*_{p^r}$.
\end{proposition}

\subsection{Codeword Length $L=(w-1)p_1^{r_1}p_2^{r_2}\cdots p_n^{r_n}$}

Given a prime $p$, an element $a\in\mathbb{Z}^*_p$ is called a \textit{quadratic residue} if there exists an integer $x\in\mathbb{Z}^*_p$ such that $a=x^2$; otherwise, $a$ is called a \textit{quadratic non-residue}.
For $a\in\Z^*_p$, denote by $\left(\frac{a}{p}\right)=1$ if $a$ is a quadratic residue modulo $p$ and $-1$ otherwise.
The notation $\left(\frac{a}{p}\right)$ is called the \textit{Legendre symbol}. 
It can be shown (e.g.,~\cite{IR90}) that the Legendre symbol is multiplicative:
\begin{equation}\label{eq:Legendre}
\left(\frac{ab}{p}\right) = \left(\frac{a}{p}\right)\left(\frac{b}{p}\right).
\end{equation}
This subsection is devoted to generalizing the following result, which is due to~\cite[Theorem 6]{LWXZ25}.

\begin{theorem}[\!\!\cite{LWXZ25}]\label{thm:single_w-1pr_old}
Let $w$ be a positive integer and $p$ be a prime such that $p\geq w$.
If 
\begin{equation*}
\left(\frac{-1}{p}\right)=-1
\end{equation*}
and
\begin{equation*}
\left(\frac{j}{p}\right)\left(\frac{j-w+1}{p}\right)=-1,\ \forall j=1,2,\ldots,w-2,
\end{equation*}
then for any integer $r\geq 1$,
\begin{align*}
K((w-1)p^r,w)=\frac{p^r-1}{2}.
\end{align*}
\end{theorem}

We first introduce the following direct construction.

\begin{theorem}\label{thm:single_w-1pr_construction}
Let $w$ be a positive integer, and let $p_1<p_2<\cdots<p_n$ be primes satisfying $p_1\geq w$.
If, for each $1\leq i\leq n$, it holds that
\begin{equation}\label{eq:Q1}
\left(\frac{-1}{p_i}\right)=-1
\end{equation}
and
\begin{equation}\label{eq:Q2}
\left(\frac{j}{p_i}\right)\left(\frac{j-w+1}{p_i}\right)=-1,\ \forall j=1,2,\ldots,w-2,
\end{equation}
then for any positive integers $r_1,\ldots,r_n$, there exists a code $\C\in\CACe(L,w)$ with 
\begin{align*}
|\C| = \frac{p_1^{r_1}\cdots p_n^{r_n}-1}{2},
\end{align*}
where $L=(w-1)p_1^{r_1}\cdots p_n^{r_n}$.
\end{theorem}
\begin{proof}
For convenience, let $L'=p_1^{r_1}\cdots p_n^{r_n}$.
Then, $L=(w-1)L'$.
Since $w-1,p_1,p_2,\ldots,p_n$ are pairwise coprime, one has
\begin{align*}
\Z_L \cong \Z_{w-1}\times\Z_{L'} \text{ and }
\Z_{L'} \cong \Z_{p_1^{r_1}}\times\cdots\times\Z_{p_n^{r_n}}.
\end{align*}
In what follows, we use ordered pairs $(z,x)\in\Z_{w-1}\times\Z_{L'}$ to represent elements in $\Z_L$, where $x\in\Z_{L'}$ are further represented by $n$-tuples $(x_1,\ldots,x_n)\in\Z_{p_1^{r_1}}\times\cdots\times\Z_{p_n^{r_n}}$.

Let $Q_i$ be the set of quadratic residues modulo $p_i$. 
For $i=1,2,\ldots,n$ define 
\begin{equation}\label{eq:single_w-1pr_generator_i}
    \begin{split}
        \widehat{Q}_i \triangleq \{(0,\ldots,0,g,x_{i+1},\ldots,x_n) &\in\Z_{p_1^{r_1}}\times\cdots\times\Z_{p_n^{r_n}}:\, \\ 
        & g\in\mathcal{S}_{r_i}(Q_i), x_t\in\Z_{p_t^{r_t}} \text{ for } i<t\leq n\}.
    \end{split}
\end{equation}
Let 
\begin{equation}\label{eq:single_w-1pr_generator}
\widehat{Q}\triangleq\biguplus_{i=1}^{n}\widehat{Q}_i
\end{equation}
be the disjoint union of $\widehat{Q}_1,\widehat{Q}_2,\ldots,\widehat{Q}_n$.
As $|Q_i|=(p_i-1)/2$, by \eqref{eq:S_r(A)-size}, we have
\begin{align*}
|\widehat{Q}| = \sum_{i=1}^{n}\frac{p_i^{r_i}-1}{2}\prod_{j=i+1}^{n}p_{j}^{r_j} = \frac{p_1^{r_1}\cdots p_n^{r_n}-1}{2}.
\end{align*}

For $a\in\widehat{Q}$, define 
\begin{equation}\label{eq:single_w-1pr_Sa}
S_a=\{j(1,a)\in\Z_{w-1}\times\Z_{L'}:\,j=0,1,2\ldots,w-1\}
\end{equation}
be the $w$-subset generated by $a$ whose set of nonzero differences is of the form
\begin{align}
d^*(S_a) &= \{\pm j(1,a)\in\Z_{w-1}\times\Z_{L'}:\,j=1,2,\ldots,w-1\} \label{eq:single_w-1pr_d*Sa} \\
&= \{\pm j(1,a)\in\Z_{w-1}\times\Z_{L'}:\,j=1,2,\ldots,w-2\} \cup \{(0,\pm (w-1)a)\}. \label{eq:single_w-1pr_d*Sa_2}
\end{align}
Assume $a=(0,\ldots,0,g,x_{i+1},\ldots,x_n)\in\Z_{p_1}^{r_1}\times\cdots\times\Z_{p_n}^{r_n}$, i.e., $a\in\widehat{Q}_i$, then the elements in $d^*(S_a)$ are of the form
\begin{align*}
\pm j(1,a) = \pm j(1,(0,\ldots,0,g,x_{i+1},\ldots,x_n))
\end{align*}
or 
\begin{align*}
(0,\pm (w-1)a) = (0,(0,\ldots,0,\pm(w-1)g,\pm(w-1)x_{i+1},\ldots,\pm(w-1)x_n)).
\end{align*}

In what follows, we shall prove that $\widehat{Q}$ can generate the desired code $\C$, that is, $d^*(S_a)$, $a\in\widehat{Q}$ are mutually disjoint.
Assume $a\in\widehat{Q}_i$ and $b\in\widehat{Q}_{i'}$, it is obvious that $d^*(S_a)\cap d^*(S_b)=\emptyset$ whenever $i\neq i'$.
So, it suffices to consider $a$ and $b$ are in the same set $\widehat{Q}_i$, for some $i$.
Suppose to the contrary that $d^*(S_a)\cap d^*(S_b)\neq\emptyset$, where $a=(0,\ldots,0,g,x_{i+1},\ldots,x_n)$ and $b=(0,\ldots,0,h,y_{i+1},\ldots,y_n)$ for some $g,h\in\mathcal{S}_{r_i}(Q_i)$.
Without loss of generality generality, assume $j(1,(0,\ldots,0,g,x_{i+1},\ldots,x_n))$ is one of the common elements in $d^*(S_a)\cap d^*(S_b)$, where $1\leq j\leq w-1$.
Then we have
\begin{equation}\label{eq:construction_w-1pr}
j(1,(0,\ldots,0,g,x_{i+1},\ldots,x_n)) = \pm k(1,(0,\ldots,0,h,y_{i+1},\ldots,y_n))
\end{equation}
for some $k\in\{1,2,\ldots,w-1\}$.
There are two cases.

\textit{Case 1:} $j=w-1$. 
From \eqref{eq:construction_w-1pr} we have $(w-1)g=\pm(w-1)h$, which implies $g=-h$ (mod $p_i^{r_i}$) due to $g\neq h$.
By considering the $p_i$-ary representation, both $g$ and $-h$ are in the same layer, say $\LL_t$ for some $t$, and thus $g_t=(-h)_t$ (mod $p_i$).
Note $(-h)_t=(-1)\cdot h_t$ (mod $p_i$) by Proposition~\ref{prop:p-ary}.
By~\eqref{eq:Legendre}, we have
\begin{align*}
\left(\frac{g_t}{p_i}\right) = \left(\frac{(-h)_t}{p_i}\right) = \left(\frac{-1}{p_i}\right)\left(\frac{h_t}{p_i}\right).
\end{align*}
However, $g_t$ and $h_t$ are quadratic residues modulo $p_i$ by the assumption that $g,h\in\mathcal{S}_{r_i}(Q_i)$ and the fact that $g,h\in\mathsf{L}_t$.
It further implies that $\left(\frac{-1}{p_i}\right)=1$, which contradicts the condition in~\eqref{eq:Q1}.

\textit{Case 2:} $1\leq j\leq w-2$.
From \eqref{eq:construction_w-1pr} we have $j(1,g)=\pm k(1,h)$ in $\Z_{w-1}\times\Z_{p_i^{r_i}}$, where $1\leq k\leq w-2$.
If $j(1,g)=k(1,h)$, it follows that $j=k$ and hence $g=h$, where a contradiction occurs.
If $j(1,g)=-k(1,h)$, we have $j+k=w-1$, and then get $jg=(j-w+1)h$ (mod $p_i^{r_i}$). 
By considering the $p_i$-ary representation, both $jg$ and $(j-w+1)h$ are in the same layer, say $\LL_t$ for some $t$.
Since $j\leq w-2<p_i$, we have $j=j_0$ and $(j-w+1)=(j-w+1)_0$ (mod $p_i$), namely, both $j$ and $j-w+1$ are in $\LL_0$.
It follows from Proposition~\ref{prop:p-ary} that  $g,h\in\LL_t$.
More precisely, 
\begin{align*}
g=j^{-1}\cdot jg\in\LL_t \text{ and } h=(j-w+1)^{-1}\cdot(j-w+1)h\in\LL_t.
\end{align*}
Therefore, $g_t,h_t\in Q_i$ by assumption.
By~\eqref{eq:Legendre} and Proposition~\ref{prop:p-ary}, 
\begin{align*}
& (jg)_t = ((j-w+1)h)_t \ (\bmod\, p_i) \\
\Rightarrow \ & j\cdot g_t = (j-w+1)\cdot h_t \ (\bmod\, p_i) \\
\Rightarrow \ & \left(\frac{j}{p_i}\right) \left(\frac{g_t}{p_i}\right) = \left(\frac{j-w+1}{p_i}\right) \left(\frac{h_t}{p_i}\right) \\
\Rightarrow \ & \left(\frac{j}{p_i}\right)  = \left(\frac{j-w+1}{p_i}\right),
\end{align*}
which contradicts the condition in~\eqref{eq:Q2}, and the proof is completed.
\end{proof}

\begin{remark}\label{rm:single_w-1pr_diff_union}\rm
Let $\C$ be the CAC constructed in Theorem~\ref{thm:single_w-1pr_construction}.
By~\eqref{eq:single_w-1pr_d*Sa_2}, we have $|d^*(S)|=2(w-1)$ for all $S\in\C$.
So, it costs $(w-1)(L'-1)$ distinct differences in total, where $L'=p_1^{r_1}\cdots p_n^{r_n}$.
Moreover, since $p_1\geq w$, it also reveals from \eqref{eq:single_w-1pr_d*Sa_2} that $(z,0)\notin d^*(S)$ for all $z\in\Z_{w-1}$ and $S\in\C$.
This concludes that
\begin{align*}
    \bigcup_{S\in\C} d^*(S) = \Z_{w-1}\times\Z^*_{L'}.
\end{align*}
\end{remark}

\begin{example}\label{ex:single_3_7_23}\rm
Let $w=4, n=2, p_1=7,$ and $p_2=23$.
The sets of quadratic residues modulo $7$ and $23$ are $Q_1=\{1,2,4\}$ and $Q_2=\{1,2,3,4,6,8,9,12,13,16,18\}$, respectively.
Obviously, $\left(\frac{-1}{7}\right)=\left(\frac{1}{7}\right)\left(\frac{-2}{7}\right)=-1$ and $\left(\frac{-1}{23}\right)=\left(\frac{1}{23}\right)\left(\frac{-2}{23}\right)=-1$.
By Theorem~\ref{thm:single_w-1pr_construction}, we have an equi-difference CAC of length $L=3\cdot 7^r\cdot 23^s$ and weight $4$ with $(7^r\cdot 23^s-1)/2$ codewords, for any positive integers $r$ and $s$.
When $r=s=1$, we have $\widehat{Q}_1=\{(g,x)\in\Z_7\times\Z_{23}:\,g\in Q_1, x\in\Z_{23}\}$ and $\widehat{Q}_2=\{(0,g)\in\Z_7\times\Z_{23}:\,g\in Q_2\}$.
The obtained code in $\CACe(483,4)$ has generators in $\{\theta^{-1}(1,a):\,a\in\widehat{Q}_1\cup\widehat{Q}_2\}$, where elements $(1,a)$ are considered in $\Z_3\times\Z_7\times\Z_{23}$ and the bijection $\theta:\Z_{483}\to\Z_3\times\Z_7\times\Z_{23}$ is given in~\eqref{eq:CRT-correspondence}.
The corresponding generators are listed as follows.
\begin{align*}
    \{\theta^{-1}(1,1,x):\,x\in\Z_{23}\} &= \{1, 22, 43, 64, 85, 106, 127, 148, 169, 190, 211, 232, \\
                                       & \qquad \! 253, 274, 295, 316, 337, 379, 385, 400, 421, 442, 463 \} \\
    \{\theta^{-1}(1,2,x):\,x\in\Z_{23}\} &= \{16, 37, 58, 79, 100, 121, 142, 163, 184, 205, 226, 247, \\
                                       & \qquad \! 268, 289, 310, 331, 352, 373, 394, 415, 436, 457, 478 \} \\
    \{\theta^{-1}(1,4,x):\,x\in\Z_{23}\} &= \{4, 25, 46, 67, 88, 109, 130, 151, 172, 193, 214, 235, \\
                                       & \qquad \! 256, 277, 298, 319, 340, 361, 382, 403, 424, 445, 466 \} \\
    \{\theta^{-1}(1,0,g):\,g\in Q_2\} &= \{49, 133, 154, 196, 232, 238, 259, 280, 301, 427, 469\}.
\end{align*}
\end{example}

We now show that the construction described in Theorem~\ref{thm:single_w-1pr_construction} is optimal.

\begin{theorem}\label{thm:single_w-1pr_optimal}
Let $w$ be a positive integer, and let $p_1<p_2<\cdots<p_n$ be primes satisfying $p_1\geq w$.
If the two conditions in~\eqref{eq:Q1} and \eqref{eq:Q2} hold, then for any positive integers $r_1,\ldots,r_n$, we have
\begin{align*}
K\left((w-1)p_1^{r_1}\cdots p_n^{r_n},w\right) = \frac{p_1^{r_1}\cdots p_n^{r_n}-1}{2}
\end{align*}
provided that
\begin{equation}\label{eq:single_w-1pr_optimal_condition}
\sum_{i=1}^{t}(2w-1-p_i)\leq w-1,
\end{equation}
where $p_t<2w-1\leq p_{t+1}$ for some $0\leq t\leq n$.
Note that $p_0\triangleq 1$ and $p_{n+1}\triangleq\infty$ by convention.
Note also that \eqref{eq:single_w-1pr_optimal_condition} always holds for $t=0,1$.
\end{theorem}
\begin{proof}
Let $L'=p_1^{r_1}\cdots p_n^{r_n}$.
By Theorem~\ref{thm:single_w-1pr_construction}, it suffices to show that for any code $\C\in\CAC((w-1)L',w)$, one has $|\C|\leq\frac{L'-1}{2}$.

Let $\mathcal{E}\subseteq\C$ be the collection of all exceptional codewords in $\C$.
For notational convenience, denote by $H_S=\HH(d(S))$ for $S\in\mathcal{E}$.
Pick any $S\in\mathcal{E}$.
Since $H_S$ is a subgroup of $\Z_{(w-1)L'}$, the size $|H_S|$ must divide $(w-1)L'$.
By Lemma~\ref{lem:exceptional_old}(iii), $|H_S|$ does not divide $(w-1)$.
As $|H_S|\geq 2$ by Lemma~\ref{lem:exceptional_old}(i), then $|H_S|$ is a multiple of $p_i$ for some $i$.
We consider two cases.

\textit{Case 1:} $p_1>2w-2$.
In this case we have $|H_S|\leq 2w-2<p_1$ due to Lemma~\ref{lem:exceptional_old}(i). 
This is a contradiction to $|H_S|$ a multiple of some prime factor $p_i$.
In other words, there is no exceptional codeword in this case.
Therefore, $|d^*(S)|\leq 2w-2$ for all $S\in\C$.
By the disjoint-difference-set property,
\begin{align*}
(w-1)L'-1 = |\Z^*_{(w-1)L'}| \geq \bigcup_{S\in\C}|d^*(S)| \geq (2w-2)|\C|,
\end{align*}
which implies that
\begin{align*}
|\C|\leq \left\lfloor\frac{L'-1}{2}+\frac{w-2}{2w-2}\right\rfloor = \frac{L'-1}{2}.
\end{align*}

\textit{Case 2:} $p_1\leq 2w-2$.
Since $0\in d(S)$ for $S\in\mathcal{E}$, if follows from Proposition~\ref{prop:stabilizer}(ii) that $H_S\subseteq d(S)$.
Then, $H_S\cap H_{S'}=\{0\}$ for any two distinct $S,S'\in\mathcal{E}$.
Moreover,  $\gcd(|H_S|,|H_{S'}|)=1$. Assume  on the contrary that $\gcd(|H_S|,|H_{S'}|)=d>1$. Since $L/d $ is a multiple of $L / |H_S| $ (also a multiple of $L/|H_{S'}|$), it follows from Proposition~\ref{prop:unique_subgroup} that $L/d \in H_S \cap H_{S'}$, a contradiction.
Now, assume $p_t<2w-1\leq p_{t+1}$ for some $1\leq t\leq n$.
Since $|H_S|\leq 2w-2<p_{t+1}$, for $S\in\mathcal{E}$, the prime factors of $|H_S|$ are in $\{p_1,\ldots,p_t\}$.
This concludes that $|\mathcal{E}|\leq t$ and $|H_S|$, $S\in\mathcal{E}$, are pairwise coprime.
Therefore,
\begin{equation}\label{eq:single_w-1pr_optimal_proof_1}
\sum_{S\in\mathcal{E}}(2w-1-|H_S|) \leq \sum_{i=1}^{t}(2w-1-p_i).
\end{equation}
On the other hand, by the disjoint-difference-set property,
\begin{align*}
(w-1)L'-1 &= |\Z^*_{(w-1)L'}| \geq \bigcup_{S\in\C}|d^*(S)| \\
&\geq (2w-2)(|\C|-|\mathcal{E}|) + \sum_{S\in\mathcal{E}}|d^*(S)|.
\end{align*}
Since $|H_S|\leq |d(S)|=|d^*(S)|+1$ due to Proposition~\ref{prop:stabilizer}(ii), it follows that
\begin{equation}\label{eq:single_w-1pr_optimal_proof_2}
(w-1)L'-1 \geq (2w-2)|\C|-\sum_{S\in\mathcal{E}}(2w-1-|H_S|).
\end{equation}
Combining \eqref{eq:single_w-1pr_optimal_condition} -- \eqref{eq:single_w-1pr_optimal_proof_2} yields
\begin{align*}
|\C| \leq \left\lfloor\frac{(w-1)L'+w-2}{2w-2}\right\rfloor = \frac{L'-1}{2}+\left\lfloor\frac{2w-3}{2w-2}\right\rfloor = \frac{L'-1}{2},
\end{align*}
as desired.
\end{proof}

The primes that satisfy the two conditions~\eqref{eq:Q1} and \eqref{eq:Q2} are characterized in~\cite{LWXZ25} with the help of Gauss's lemma and laws of quadratic reciprocity, which lead to a series of optimal CACs of length $(w-1)p^r$ in~\cite[Corollary 2]{LWXZ25}.
As an application of Theorem~\ref{thm:single_w-1pr_optimal}, the following corollary generalizes these optimal CACs.

\begin{corollary}\label{coro:single_w-1pr}
Let $L'=\prod_{i=1}^{n}p_i^{r_i}$, where $p_i$ are distinct primes and $r_i$ are any positive integers. 
One has
\begin{enumerate}
\item $K(3L',4)=(L'-1)/2$ if $p_i\equiv -1\ (\bmod\ 8)$ for all $i$,
\item $K(4L',5)=(L'-1)/2$ if $p_i\equiv -1\ (\bmod\ 12)$ for all $i$, 
\item $K(5L',6)=(L'-1)/2$ if $p_i\equiv -1,-5\ (\bmod\ 24)$ for all $i$, 
\item $K(6L',7)=(L'-1)/2$ if $p_i\equiv -1,-9\ (\bmod\ 40)$ for all $i$, 
\item $K(7L',8)=(L'-1)/2$ if $p_i\equiv -1,-49\ (\bmod\ 120)$ for all $i$, 
\item $K(8L',9)=(L'-1)/2$ if  $p_i\equiv -1,59,-109,-121,131,-169\ (\bmod\ 420)$ for all $i$, 
\item $K(9L',10)=(L'-1)/2$ if  $p_i\equiv -1,-9,31,-81,111,-121\ (\bmod\ 280)$ for all $i$, and
\item $K(10L',11)=(L'-1)/2$ if  $p_i\equiv -1,-5,-25,43,47,67\ (\bmod\ 168)$ for all $i$.
\end{enumerate}
\end{corollary}

Note that we can derive a sufficient condition of primes $p_i$ so that $K((w-1)L',w)=(L'-1)/2$, for any arbitrary $w$.
Please refer to~\cite{LWXZ25} for more details.

\subsection{Codeword Length $L=p_1^{r_1}p_2^{r_2}\cdots p_n^{r_n}$}

This subsection is devoted to generalizing the following result, proposed in~\cite[Theorem 10]{LWXZ25}.

\begin{theorem}[\!\!\cite{LWXZ25}]\label{thm:CAC_pr}
Let $w$ be a positive integer, and let $p$ be a prime such that $p-1$ is divided by $2w-2$.
If there is a code in $\CACe(p,w)$ with $(p-1)/(2w-2)$ codewords, then for any integer $r\geq 1$, 
\begin{align*}
K\left(p^r,w\right)=\frac{p^r-1}{2w-2}.
\end{align*}
\end{theorem}

We first introduce the following construction.

\begin{theorem}\label{thm:single_pr_construction}
Let $w$ be a positive integer and let $p_1<p_2<\cdots<p_n$ be primes that satisfy $p_1\geq 2w-1$.
If there is a code in $\CACe(p_i,w)$ with $m_i$ codewords for each $i$, then for any positive integers $r_1,\ldots,r_n$, there exists a code $\C\in\CACe(L,w)$ with 
\begin{equation}\label{eq:single_pr_construction_size}
|\C| = \sum_{i=1}^{n}\frac{m_i(p_i^{r_i}-1)}{p_i-1}\prod_{j=i+1}^{n}p_j^{r_j},
\end{equation}
where $L=p_1^{r_1}\cdots p_n^{r_n}$.
\end{theorem}
\begin{proof}
Since $p_1,p_2,\ldots,p_n$ are distinct primes, 
\begin{align*}
\Z_L \cong \Z_{p_1^{r_1}}\times\Z_{p_2^{r_2}}\times\cdots\times\Z_{p_n^{r_n}}.
\end{align*}
In what follows, we use $n$-tuples $(x_1,\ldots,x_n)\in\Z_{p_1^{r_1}}\times\cdots\times\Z_{p_n^{r_n}}$ to represent elements in $\Z_L$.

For $i=1,2,\ldots, n$, let $\Gamma_i$ be a set of $m_i$ generators of the given code in $\CACe(p_i, w)$. 
Then, define
\begin{align*}
\widehat{\Gamma_i}\triangleq\{(0,\ldots,0,g,x_{i+1},\ldots,x_n) &\in \Z_{p_1^{r_1}}\times \cdots \times \Z_{p_n^{r_n}}:\, \\ 
& g\in\mathcal{S}_{r_i}(\Gamma_i), x_j\in\Z_{p_j^{r_j} }\text{ for }i<j\leq n\}.
\end{align*}
It is easy to see from \eqref{eq:S_r(A)-size} that for $i=1,2,\ldots,n$,
\begin{align*}
|\widehat{\Gamma}_i| = \frac{m_i(p_i^{r_i}-1)}{p_i-1} \prod_{j=i+1}^{n}p_j^{r_j}.
\end{align*}
Note that the sets $\widehat{\Gamma}_1,\widehat{\Gamma}_2,\ldots,\widehat{\Gamma}_n$ are mutually disjoint.
Let  
\begin{equation}\label{eq:single_pr_generator}
\widehat{\Gamma}=\biguplus_{i=1}^{n}\widehat{\Gamma}_i
\end{equation}
be the disjoint union of these sets.
Obviously,
\begin{align*}
|\widehat{\Gamma}|=\sum_{i=1}^{n}\frac{m_i(p_i^{r_i}-1)}{p_i-1} \prod_{j=i+1}^{n}p_j^{r_j}.
\end{align*}

Consider the set $\widehat{\Gamma}$.
For $a\in\widehat{\Gamma}$, if $a=(0,\ldots,0,g,x_{i+1},\ldots,x_{n})$ for some $g\in\mathcal{S}_{r_i}(\Gamma_i)$, i.e., $a\in\widehat{\Gamma}_i$, then define
\begin{equation}\label{eq:sginel_pr_Sa}
S_a=\{ja\in\Z_{p_1^{r_1}}\times \cdots \times \Z_{p_n^{r_n}} :\, j=0,1,2,\ldots,w-1\}
\end{equation}
be the $w$-subset generated by $a$ whose set of nonzero differences is of the form
\begin{equation}\label{eq:single_pr_d*Sa}
d^*(S_a)=\{ja\in\Z_{p_1^{r_1}}\times \cdots \times \Z_{p_n^{r_n}} :\, j=\pm 1,\pm 2,\ldots,\pm(w-1)\}.
\end{equation} 
In what follows, we shall prove that $\widehat{\Gamma}$ is the set of generators of the desired code $\C$, that is, $d^*(S_a)\cap d^*(S_b)=\emptyset$ for any distinct $a,b\in\widehat{\Gamma}$.
It suffices to consider that $a$ and $b$ are in the same subset $\widehat{\Gamma}_i$, for some $i$, because the assertion is definitely true in the other case.

Suppose to the contrary that $d^*(S_a)\cap d^*(S_b)\neq\emptyset$, where $a=(0,\ldots,0,g,x_{i+1},\ldots,x_n)$ and $b=(0,\ldots,0,h,y_{i+1},\ldots,y_n)$ for some $g,h\in\mathcal{S}_{r_i}(\Gamma_i)$.
Assume $j(0,\ldots,0,g,x_{i+1},\ldots,x_n)=k(0,\ldots,0,h,y_{i+1},\ldots,y_n)$ for some $j,k\in\{\pm 1, \pm 2, \ldots, \pm(w_i-1)\}$.
If $j=k$, then $g=h$ and $x_t=y_t$ for $i+1\leq t\leq n$, namely, $a=b$.
So, we consider $j\neq k$ in what follows.
Let us focus on the $i$-th component, namely, $jg=kh$.
Since $p_i\geq 2w-1$, by Proposition~\ref{prop:p-ary}, one has $jg\in\LL_t$ for every $j\in\{\pm1,\pm2,\ldots,\pm(w-1)\}$ if $g\in\LL_t$.
Here we use $\LL_t$ to denote the $t$-th layer in the $p_i$-ary representation of elements of $\Z_{p_i^{r_i}}$ for notatinoal convenience.
It follows that $jg\neq kh$ for any $j,k\in\{\pm1,\pm2,\ldots,\pm(w-1)\}$ whenever $g$ and $h$ are in distinct layers.
Therefore, we only need to consider the case when both $g$ and $h$ are in the same layer, say $\LL_t$ for some $t$.
Observe that $j,k\in\LL_0$ due to $p_i\geq 2w-1$. 
By Proposition~\ref{prop:p-ary} again, $j\cdot g_t=k\cdot h_t$ (mod $p_i$).
If $g_t\neq h_t$, by the assumption that $g_t,h_t\in\Gamma_i$ be two distinct generators in the given code in $\CACe(p_i,w)$, one has $j\cdot g_t\neq k\cdot h_t$.
If $g_t=h_t$, it further implies that $(j-k)g_t=0$ (mod $p_t$), which is impossible because of $j,k$ being distinct in $\{\pm 1, \pm2,\cdots,\pm (w-1)\}$ and $p_i\ge 2w-1$.
This completes the proof.
\end{proof}

\begin{example}\label{ex:single_7_37}\rm
Let $w=4, n=2, p_1=17$, and $p_2=37$.
One can check that $\Gamma_1=\{1,4\}$ forms a set of generators of a code in $\CACe(17,4)$ and $\Gamma_2=\{1,8,23,26,27,31\}$ forms a set of generators of a code in $\CACe(37,4)$.
By Theorem~\ref{thm:single_pr_construction}, we have an equi-difference CAC of length $L=17^r\cdot 37^s$ and weight $4$ with $\frac{1}{8}(17^r-1)37^s+\frac{1}{6}(37^s-1)$ codewords, for any positive integers $r$ and $s$.
When $r=s=1$, we have $\widehat{\Gamma}_1=\{(g,x)\in\Z_{17}\times\Z_{37}:\,g\in\Gamma_1, x\in\Z_{37}\}$ and $\widehat{\Gamma}_2=\{(0,g)\in\Z_{17}\times\Z_{37}:\,g\in\Gamma_2\}$.
The obtained code in $\CACe(629,4)$ has generators in $\{\theta^{-1}(a):\,a\in\widehat{\Gamma}_1\cup\widehat{\Gamma}_2\}$, where the bijection $\theta:\Z_{629}\to\Z_{17}\times\Z_{37}$ is given in~\eqref{eq:CRT-correspondence}.
The generators obtained from $\{(1,x):\,x\in\Z_{37}\}$ are 1, 18, 35, 52, 69, 86, 103, 120, 137, 154, 171, 188, 205, 222, 239, 256, 273, 290, 307, 324, 341, 358, 375, 392, 409, 426, 443, 460, 477, 494, 511, 528, 545, 562, 579, 596, and 613.
The generators from $\{(4,x):\,x\in\Z_{37}\}$ are given by adding $3$ modulo $629$ to those obtained from  $\{(1,x):\,x\in\Z_{37}\}$.
And, the generators from $\widehat{\Gamma}_2$ are 68, 119, 323, 408, 544, 578.
The resulting CAC contains $80$ codewords.
\end{example}



We now show that the construction described in Theorem~\ref{thm:single_pr_construction} is optimal in some cases.

\begin{theorem}\label{thm:single_pr_optimal}
Let $w$ be a positive integer, and let $p_1<\cdots<p_n$ be primes such that $p_i-1$ is divided by $2w-2$ for each $i$.
If there is a code in $\CACe(p_i,w)$ with $(p_i-1)/(2w-2)$ codewords for each $i$, then for any positive integers $r_1,\ldots,r_n$, 
\begin{equation}\label{eq:single_pr_optimal}
K(p_1^{r_1}\cdots p_n^{r_n},w)= \frac{p_1^{r_1}\cdots p_n^{r_n}-1}{2(w-1)}.
\end{equation}
\end{theorem}
\begin{proof}
The assumption that $p_i-1$ is divided by $2w-2$ guarantees $p_i\geq 2w-1$, for each $i$.
By plugging $m_i=(p_i-1)/(2w-2)$ into \eqref{eq:single_pr_construction_size}, there exists a code in $\CACe(p_1^{r_1}\cdots p_n^{r_n},w)$ with $(p_1^{r_1}\cdots p_n^{r_n}-1)/(2w-2)$ codewords.
Since the least prime factor $p_1$ is larger than or equal to $2w-1$, the result follows by Theorem~\ref{thm:UB-single}.
\end{proof}

Note that equi-difference CACs of weight $w$ and length a prime $p$ with $(p-1)/(2w-2)$ codewords are studied in~\cite{MMSJ07,Momihara07,FLS14,ZS22}.

\begin{remark}\label{rm:single_pr_diff_union}\rm
Let $\C$ be an optimal CAC considered in Theorem~\ref{thm:single_pr_optimal}.
Since $p_1 \geq 2w-1$, it follows from~\eqref{eq:single_pr_d*Sa} that $|d^*(S)|=2(w-1)$ for $S\in\C$.
As $|\C|=(L-1)/2$, where $L=p_1^{r_1}\cdots p_n^{r_n}$, we have
\begin{equation*}
    \bigcup_{S\in\C} d^*(S) = \Z^*_L.
\end{equation*}
So, the optimal CACs obtained in Theorem~\ref{thm:single_pr_construction} is \emph{tight}~\cite{Momihara07}.
\end{remark}

%
%


\section{Multichannel CACs with Two Channels}\label{sec:MC-CAC_2}

By the help of constructions given in Theorems~\ref{thm:single_w-1pr_construction} and \ref{thm:single_pr_construction}, we have the following construction for MC-CACs with $M=2$ channels.

\begin{theorem}\label{thm:2channel_w-1pr_construction}
Let $w$ be a positive integer, and let $p_1<p_2<\cdots<p_n$ be primes satisfying $p_1\geq 2w-1$.
For each $1\leq i\leq n$, if there is a code in $\CACe(p_i,w)$ with $m_i$ codewords and the two conditions in~\eqref{eq:Q1} and \eqref{eq:Q2} hold, then for any positive integers $r_1,\ldots,r_n$, there exists a code $\C\in\MCCAC(2,L,w)$ with
\begin{equation}\label{eq:2channel_w-1pr_construction_size}
|\C| = \prod_{i=1}^{n}p_i^{r_i} + \sum_{i=1}^{n} \frac{2m_i(p_i^{r_i}-1)}{p_i-1}\prod_{j=i+1}^{n}p_j^{r_j}
\end{equation}
codewords, where $L=(w-1)p_1^{r_1}p_2^{r_2} \cdots p_n^{r_n}$. 
\end{theorem}
\begin{proof}
Let $L'=p_1^{r_1}p_2^{r_2} \cdots p_n^{r_n}$, then $L=(w-1)L'$.
In the following construction, the desired code $\C$ will consist of three classes of codewords $\C_1$, $\C_2$ and $\C_3$ with sizes
\begin{align*}
|\C_1|=L'-1, \ |\C_2|=\sum_{i=1}^{n} \frac{2m_i(p_i^{r_i}-1)}{p_i-1}\prod_{j=i+1}^{n}p_j^{r_j}, \ \text{and } |\C_3|=1,
\end{align*}
where the codewords in $\C_1$ and $\C_2$ are from the constructions in Theorem~\ref{thm:single_w-1pr_construction} and Theorem~\ref{thm:single_pr_construction}, respectively.

Since $\gcd(w-1,L')=1$, we have $\Z_L\cong\Z_{w-1}\times\Z_{L'}$.
Elements in $\mathcal{I}_2\times\Z_L$ can be represented by $(m,(z,x))$, where $m\in\mathcal{I}_2=\{1,2\}$ and $z\in\Z_{w-1}$, $x\in\Z_{L'}$.

Since the two conditions \eqref{eq:Q1} -- \eqref{eq:Q2} hold, Theorem~\ref{thm:single_w-1pr_construction} ensures the existence of a code in $\CAC(L,w)$ with $(L'-1)/2$ codewords.
We adopt the notation in the proof of Theorem~\ref{thm:single_w-1pr_construction}.
Consider the set $\widehat{Q}$ defined in~\eqref{eq:single_w-1pr_generator}.
Note that $\widehat{Q}\subset\Z_{L'}$ with size $|\widehat{Q}|=(L'-1)/2$.
For $a\in\widehat{Q}$ define the following two codewords in $\mathcal{I}_2\times\Z_L$
\begin{align*}
S_a^1 &= \{(1, -(1,a))\} \cup \{(2,j(1,a)):\,j=0,1,2,\ldots,w-2\}, \text{ and } \\
S_a^2 &= \{(1,j(1,a)):\,j=0,1,2,\ldots,w-2\} \cup \{(2, -(1,a))\}.
\end{align*}
Let $\C_1=\{S_a^1,S_a^2:\,a\in\widehat{Q}\}$ be the collection of these codewords derived from all elements in $\widehat{Q}$.
Let us focus on the arrays of differences $D_{S_a^1}$ and $D_{S_a^2}$.
It is easy to see that
\begin{align}
D_{S_a^1}(1,1) = D_{S_a^2}(2,2) &= \emptyset, \label{eq:2channel_w-1pr_C1_DS-1} \\
D_{S_a^1}(2,2) = D_{S_a^2}(1,1) &= \{\pm j(1,a)\in\Z_{w-1}\times\Z_{L'}:\,j=1,2,\ldots,w-2\}, \label{eq:2channel_w-1pr_C1_DS-2} \\
D_{S_a^1}(1,2) = D_{S_a^2}(2,1) &= \{-j(1,a)\in\Z_{w-1}\times\Z_{L'}:\,j=1,2,\ldots,w-1\}, \label{eq:2channel_w-1pr_C1_DS-3} \\
D_{S_a^1}(2,1) = D_{S_a^2}(1,2) &= \{j(1,a)\in\Z_{w-1}\times\Z_{L'}:\,j=1,2,\ldots,w-1\}. \label{eq:2channel_w-1pr_C1_DS-4}
\end{align}
Recall that in the proof of Theorem~\ref{thm:single_w-1pr_construction}, it was shown that $d^*(S_a)\cap d^*(S_b)=\emptyset$ for any two distinct $a,b\in\widehat{Q}$, where $d^*(S_a)$ is given in~\eqref{eq:single_w-1pr_d*Sa} as
\begin{align*}
d^*(S_a)=\{\pm j(1,a)\in\Z_{w-1}\times\Z_{L'}:\,j=1,2,\ldots,w-1\}.
\end{align*}
Therefore, it follows from \eqref{eq:2channel_w-1pr_C1_DS-1} -- \eqref{eq:2channel_w-1pr_C1_DS-4} that $D_{S_a^1}, D_{S_a^2}$ are mutually disjoint over all $a\in\widehat{Q}$.
Moreover, by the property given in Remark~\ref{rm:single_w-1pr_diff_union}, we further derive that
\begin{align}
\bigcup_{S\in\C_1} D_S = 
   \begin{tabular}{|c|c|} \hline
$\Z^*_{w-1}\times\Z^*_{L'}$ & $\Z_{w-1}\times\Z^*_{L'}$ \\ \hline
$\Z_{w-1}\times\Z^*_{L'}$ & $\Z^*_{w-1}\times\Z^*_{L'}$ \\ \hline
\end{tabular} .
\label{eq:2channel_w-1pr_C1_DS}
\end{align}

Since there is a code in $\CACe(p_i,w)$ with $m_i$ codewords for each $i$, it follows from Theorem~\ref{thm:single_pr_construction} that there exists a code in $\CAC(L',w)$ with $\sum_{i=1}^{n} \left(\frac{m_i(p_i^{r_i}-1)}{p_i-1}\prod_{j=i+1}^{n}p_j^{r_j}\right)$ codewords, where the generators are collected in the set $\widehat{\Gamma}$ given in~\eqref{eq:single_pr_generator}.
We adopt the notation in the proof of Theorem~\ref{thm:single_pr_construction}.
For $a\in\widehat{\Gamma}$ define the following two codewords in $\mathcal{I}_2\times\Z_{L'}$
\begin{align*}
T_a^1 &= \{(1,(0,ja)):\,j=0,1,2,\ldots,w-1\}, \text{ and } \\
T_a^2 &= \{(2,(0,ja)):\,j=0,1,2,\ldots,w-1\}.
\end{align*}
Let $\C_2=\{T_a^1,T_a^2:\,a\in\widehat{\Gamma}\}$ be the collection of these codewords.
Consider the arrays of differences $D_{T_a^1}$ and $D_{T_a^2}$.
It is obvious that
\begin{align}
& D_{T_a^1}(1,2)=D_{T_a^1}(2,1)=D_{T_a^1}(2,2)=D_{T_a^2}(1,2)=D_{T_a^2}(2,1)=D_{T_a^2}(1,1)=\emptyset, \text{ and} \label{eq:2channel_w-1pr_C2_DS-1} \\
& D_{T_a^1}(1,1) = D_{T_a^2}(2,2) = \{(0,\pm ja)\in\Z_{w-1}\times\Z_{L'}:\,j=1,2,\ldots,w-1\}. \label{eq:2channel_w-1pr_C2_DS-2}
\end{align}
In the proof of Theorem~\ref{thm:single_pr_construction}, it was shown that $\{\pm ja\in\Z_{L'}:\,j=1,2,\ldots,w-1\}$, $a\in\widehat{\Gamma}$, are mutually disjoint.
Then, it follows from \eqref{eq:2channel_w-1pr_C2_DS-1} -- \eqref{eq:2channel_w-1pr_C2_DS-2} that $D_{T_a^1}, D_{T_a^2}$ are mutually disjoint over all $a\in\widehat{\Gamma}$.
Moreover, we have
\begin{align}
\bigcup_{S\in\C_2} D_S \subseteq 
   \begin{tabular}{|c|c|} \hline
$\{0\}\times\Z^*_{L'}$ & $\emptyset$ \\ \hline
$\emptyset$ & $\{0\}\times\Z^*_{L'}$ \\ \hline
\end{tabular} .
\label{eq:2channel_w-1pr_C2_DS}
\end{align}

Finally, the third class $\C_3$ only contains one codeword given by
\begin{align*}
U = \{(1,(j,0)):\,j=0,1,\ldots,w-2\}\cup\{(2,(0,0))\}.
\end{align*}
Observe that 
\begin{align}
D_U = 
\begin{tabular}{|c|c|} \hline
$\Z^*_{w-1}\times\{0\}$ & $\Z_{w-1}\times\{0\}$ \\ \hline
$\Z_{w-1}\times\{0\}$ & $\emptyset$ \\ \hline
\end{tabular} .
\label{eq:2channel_w-1pr_C3_DS}
\end{align}
By \eqref{eq:2channel_w-1pr_C1_DS}, \eqref{eq:2channel_w-1pr_C2_DS}, \eqref{eq:2channel_w-1pr_C3_DS} and aforementioned arguments, $\C_1\cup\C_2\cup\C_3$ forms the desired MC-CAC.
\end{proof}

\begin{example}\label{ex:2channel_3_7_23} \rm
Let $w=4, n=2,p_1=7$, and $p_2=23$.
We have $L=483$.
It has been shown in Example~\ref{ex:single_3_7_23} that $7$ and $23$ satisfy the two conditions in~\eqref{eq:Q1} and \eqref{eq:Q2}.
It is easy to see that $\Gamma_1=\{1\}$ generates a code in $\CACe(7,4)$ and $\Gamma_2=\{1,5\}$ generates a code in $\CACe(23,4)$.
By Theorem~\ref{thm:2channel_w-1pr_construction}, we have a code in $\MCCAC(2,3\cdot 7^r\cdot 23^s, 4)$ with $7^r\cdot 23^s+\frac{1}{3}(7^r-1)23^s+\frac{2}{11}(23^s-1)$ codewords, for any positive integers $r$ and $s$.

Consider $r=s=1$. 
The first $7\cdot 23-1$ codewords (collected in $\C_1$) are those in the form 
\begin{align*}
&\{(1,-(1,a)),(2,(0,0)),(2,(1,a)), (2,2(1,a))\} \quad \text{and} \\
&\{(1,(0,0)),(1,(1,a)),(1,2(1,a)), (2,-(1,a))\}
\end{align*}
for all elements $a$ in $\widehat{Q}$, which is given in Example~\ref{ex:single_3_7_23}, and the ordered pairs $(1,a), -(1,a)$ are considered in $\Z_3\times(\Z_7\times\Z_{23})$.
Take $a=(1,21)\in\widehat{Q}_1$ as an example.
As $a=113\in\Z_{161}$, $(1,a)=274\in\Z_{483}$ and $-(1,a)=109\in\Z_{483}$ by the CRT correspondence, the two corresponding codewords are $\{(1, 109), (2,0), (2,274), (2,65)\}$ and $\{(1,0), (1,274), (2,65), (2, 109)\}$.

The $50$ codewords in $\C_2$ are those in the form 
\begin{align*}
    \{(m,(0,0)),(m,(0,a)),(m,(0,2a)),(m,(0,3a))\}
\end{align*}
for $m=1,2$ and elements $a\in\widehat{\Gamma}=\{(1,x)\in\Z_7\times\Z_{23}:\,x\in\Z_{23}\}\cup\{(0,1),(0,4)\}$.
Take $a=(1,21)$ as an example.
As $(0,a)=435\in\Z_{483}$ by the CRT correspondence, the two corresponding codewords are $\{(1,0),(1,435),(1,387),(1,339)\}$ and $\{(2,0),(2,435),(2,387),(2,339)\}$.
Finally, the last codeword (in $\C_3$) is $\{(1,0),(1,322),(1,161),(2,0)\}$.
\end{example}

We have the following upper bound on the value $K(2,L,w)$ when the length $L$ is a product of larger prime factors.

\begin{theorem}\label{thm:2channel_w-1pr_upper}
Let $L'$ and $w$ be positive integers.
If the prime factors of $L'$ are all larger than or equal to $2w-1$, then
\begin{align*}
    K(2,(w-1)L',w) \leq L' + \left\lfloor\frac{L'+w-3}{w-1}\right\rfloor.
\end{align*}
\end{theorem}
\begin{proof}
Let $L=(w-1)L'$, and let $\C$ be any code in $\MCCAC(2,L,w)$. 
We first partition codewords in $\C$ into two classes, $\C_1$ and $\C_2$, by
\begin{align*}
\C_1 &= \{S\in \C:\, e_S=1\},\\
\C_2 &= \{S\in \C:\, e_S=2\}.
\end{align*}

Consider the codewords in $\C_1$.
Suppose $S\in\C_1$ is exceptional.
By Lemma~\ref{lem:exceptional_old}(i), we have $|\HH(S-S)|\leq 2w-2$, which is strictly less than any factor of $L'$.
Note that $|\HH(S-S)|$ divides $L=(w-1)L'$ as $\HH(S-S)$ is a subgroup of $\Z_L$ due to Proposition~\ref{prop:stabilizer}(i).
It must be the case that $|\HH(S-S)|$ divides $w-1$.
This is a contradiction to Lemma~\ref{lem:exceptional_old}(iii).
Therefore, there is no exceptional codeword in $\C_1$.
It follows that
\begin{equation}\label{eq:2channel_w-1pr_optimal_1}
\sum_{S\in\C_1}|D_S(1,1)|+|D_S(2,2)| \geq (2w-2)|\C_1|.
\end{equation}

Now, consider the codewords in $\C_2$.
If $\{S_1,S_2\}$ is an exceptional pair, by Lemma~\ref{lem:exceptional_new}(i), we have $|\HH(S_1-S_2)|\leq |S_1|+|S_2|-2=w-2$, which is strictly less than any factor of $L'$.
As $\HH(S_1-S_2)$ is a subgroup of $\Z_L$, it implies that $|\HH(S_1-S_2)|$ divides $w-1$, which is a contradiction to Lemma~\ref{lem:exceptional_new}(ii).
So, $\{S_1,S_2\}$ is not an exceptional pair.
Hence, for $S\in\C_2$, one has $|D_S(1,2)|=|S_1-S_2|\geq |S_1|+|S_2|-1=w-1$.
It follows that
\begin{equation}\label{eq:2channel_w-1pr_optimal_2}
(w-1)L' \geq \sum_{S\in\C}|D_S(1,2)| = \sum_{S\in\C_2}|D_S(1,2)| \geq (w-1)|\C_2|.
\end{equation}

Following the arguments above, if $S$ is an exceptional codeword in $\C_2$, it must be the case that either $S_1$, $S_2$, or both are exceptional.
To make a more accurate estimation of the number of differences that are derived from all exceptional codewords, we define, for $i=1,2$,  
\begin{align*}
\C_2^{(i)}\triangleq\{S_i:\, S\in\C_2\}
\end{align*}
and
\begin{align*}
\mathcal{E}^{(i)}\triangleq\{S_i\in\C_2^{(i)}:\,S_i\text{ is exceptional}\}.
\end{align*}
Note that $|\C_2^{(1)}|=|\C_2^{(2)}|=|\C_2|$ and
\begin{equation}\label{eq:2channel_w-1pr_optimal_decomposeC2}
\sum_{A\in\C_2^{(1)}}|A|+\sum_{B\in\C_2^{(1)}}|B|=w|\C_2|.
\end{equation}

Consider an element $A\in\mathcal{E}^{(1)}$.
For notational convenience, denote by $H_A=\HH(A-A)$.
Since, by Proposition~\ref{prop:stabilizer}(i), $H_A$ contains the element $0$, so we can further denote by $H_A^*=H_A\setminus\{0\}$.
By the same argument given above, $|H_A|$ must divide $w-1$.
As $H_A$ is a subgroup of $\Z_L$, we have $H_A=-H_A$, which implies that $|-A+H_A|=|-(A+H_A)|=|A+H_A|$.
By plugging $A=A$ and $B=-A$ into \eqref{eq:Kneser1}, we have
\begin{align*}
|d(A)|=|A+(-A)| &\geq |A+H_A|+|-A+H_A|-|H_A| \\
&= 2|A+H_A|-|H_A| \geq 2|A|-|H_A^*|-1,
\end{align*}
which leads to
\begin{equation}\label{eq:2channel_w-1pr_optimal_3}
|d^*(A)| \geq 2|A|-|H^*_A|-2.
\end{equation}
Since $H_A$ is a subgroup of $\Z_L$ and $|H_A|$ divides $w-1$, by Proposition~\ref{prop:unique_subgroup}, $H_A$ is a subgroup of $G=\{iL':\, i=0,1,\ldots,w-2\}$.
Moreover, observe that $H_A\subseteq d(A)$ due to Proposition~\ref{prop:stabilizer}(ii).
Then, by the definition of an MC-CAC, $H^*_A\cap H^*_{A'}=\emptyset$ for any pair of distinct $A,A'\in\mathcal{E}^{(1)}$.
This concludes that
\begin{align*}
\sum_{A\in\mathcal{E}^{(1)}}|H^*_A| = \left|\biguplus_{A\in\mathcal{E}^{(1)}}H^*_A\right| \leq |G\setminus\{0\}| = w-2.
\end{align*}
By~\eqref{eq:2channel_w-1pr_optimal_3}, we further have
\begin{align*}
\sum_{A\in\mathcal{E}^{(1)}}|d^*(A)| \geq \left(\sum_{A\in\mathcal{E}^{(1)}}2|A|-2\right) - (w-2).
\end{align*}
Since $|d^*(A)|\geq 2|A|-2$ if $A$ is not exceptional, it follows that
\begin{equation}\label{eq:2channel_w-1pr_optimal_4}
\sum_{A\in\C_2^{(1)}}|d^*(A)| \geq \left(\sum_{A\in\C_2^{(1)}}2|A|-2\right) - (w-2).
\end{equation}
Similarly, we have
\begin{equation}\label{eq:2channel_w-1pr_optimal_5}
\sum_{B\in\C_2^{(2)}}|d^*(B)| \geq \left(\sum_{B\in\C_2^{(2)}}2|B|-2\right) - (w-2).
\end{equation}

By \eqref{eq:2channel_w-1pr_optimal_1}, \eqref{eq:2channel_w-1pr_optimal_decomposeC2}, \eqref{eq:2channel_w-1pr_optimal_4} and \eqref{eq:2channel_w-1pr_optimal_5}, 
\begin{align*}
2(w-1)L'-2 = 2|\Z^*_L| &\geq \sum_{S\in\C}|D_S(1,1)|+\sum_{S\in\C}|D_S(2,2)| \\
&= \sum_{S\in\C_1}\left(|D_S(1,1)|+|D_S(2,2)|\right) + \sum_{S\in\C_2}\left(|D_S(1,1)|+|D_S(2,2)|\right) \\
&\geq (2w-2)|\C_1| + \sum_{A\in\C_2^{(1)}}|d^*(A)| + \sum_{B\in\C_2^{(1)}}|d^*(B)| \\
&\geq (2w-2)|\C_1| + (2w-4)|\C_2| - 2(w-2),
\end{align*}
which yields
\begin{equation}\label{eq:2channel_w-1pr_optimal_6}
(w-1)L'-1 \geq (w-1)|\C_1| + (w-2)|\C_2| - (w-2).
\end{equation}
Combining it with~\eqref{eq:2channel_w-1pr_optimal_2}, we obtain
\begin{align*}
wL'-1 \geq (w-1)|\C|-(w-2),
\end{align*}
and hence the result follows.
\end{proof}

We now show that the construction described in Theorem~\ref{thm:2channel_w-1pr_construction} is optimal in some cases.

\begin{theorem}\label{thm:2channel_w-1pr_optimal}
Let $w$ be a positive integer, and let $p_1<\cdots<p_n$ be primes such that $p_i-1$ is divisible by $2w-2$ for each $i$.
For each $1\leq i\leq n$, if there is a code in $\CACe(p_i,w)$ with $(p_i-1)/(2w-2)$ codewords and the two conditions in~\eqref{eq:Q1} and \eqref{eq:Q2} hold, then for any positive integers $r_1,\ldots,r_n$, 
\begin{align*}
K(2,L,w) = \prod_{i=1}^{n}p_i^{r_i} + \frac{\left(\prod_{i=1}^{n}p_i^{r_i}\right)-1}{w-1},
\end{align*}
where $L=(w-1)p_1^{r_1}\cdots p_n^{r_n}$.
\end{theorem}
\begin{proof}
The assumption that $p_1-1$ is divisible by $2w-2$ guaranties $p_1\ge 2w-1\geq w$. 
By plugging $m_i=(p_i-1)/(2w-2)$ into Theorem \ref{thm:2channel_w-1pr_construction}, there exists a code in $\MCCAC(2,L,w)$ with $\prod_{i=1}^{n}p_i^{r_i} + \left(\left(\prod_{i=1}^{n}p_i^{r_i}\right)-1\right)/(w-1)$ codewords.

Let $\C$ be any code in $\MCCAC(2,L,w)$.
By plugging $L'=\prod_{i=1}^{n}p_i^{r_i}$ into Theorem~\ref{thm:2channel_w-1pr_upper},
\begin{align*}
|\C| &\leq \prod_{i=1}^{n}p_i^{r_i}+\frac{\left(\prod_{i=1}^{n}p_i^{r_i}\right)-1}{w-1} +\left\lfloor\frac{w-2}{w-1} \right\rfloor \\
&= \prod_{i=1}^{n}p_i^{r_i}+\frac{\left(\prod_{i=1}^{n}p_i^{r_i}\right)-1}{w-1}.
\end{align*}
This completes the proof.
\end{proof}

For $w=4$, the primes that satisfy the conditions given in Theorem~\ref{thm:2channel_w-1pr_optimal} are 7, 607, 631, 751, 919, 1087, 1447, 2239, 2287, 2311, 2647, 2719, and so on.

\section{Multichannel CACs with $M>2$ Channels}\label{sec:MC-CAC_M}

For $M>2$ channels, we only focus on the scenario that weight $w$ is divisible by $M$.
We first propose the following consruction.

\begin{theorem}\label{thm:Mchannel_construction}
Let $M,L,L'$, and $w$ be positive integers such that $w>M\geq 3$, $M$ divides $w$, $L=(2\frac{w}{M}-1)L'$, and all prime factors of $L'$ are larger than or equal to $2w-1$. 
If there is a code in $\CAC(L',w)$ with $h$ codewords, then there exists a code $\C\in \MCCAC(M, L, w)$ with $|\C|=Mh+L'$ codewords.
\end{theorem}
\begin{proof}
For convenience, let $t=\frac{w}{M}$.
Then $L=(2t-1)L'$.
Since every prime factor $p$ of $L'$ satisfies $p\ge 2w-1>2t-1$, then $\gcd(2t-1, L')=1$, thus $\Z_L \cong \Z_{2t-1} \times \Z_{L'}$.
So, for the sake of conveniences, the elements in $\Z_L$ can be represented as ordered pairs in $\Z_{2t-1}\times \Z_{L'}$. 

The objective code will be decomposed into two classes, $\C_1$ and $\C_2$, whose constructions are given below.

Let $\C'$ be the code in $\CAC(L',w)$ with $|\C'|=h$ codewords. 
For $S\in\C'$ and $m\in\mathcal{I}_M$, define a $w$-subset in $\mathcal{I}_M\times\Z_L$ by
\begin{equation}\label{eq:M-channel_construction_C1}
\widehat{S}_m \triangleq \{(m,(0,a))\in\mathcal{I}_M\times\Z_L:\,a\in S\}.
\end{equation}
Let 
\begin{align*}
\C_1=\bigcup_{S\in\C',\, m\in\mathcal{I}_M} \widehat{S}_m.
\end{align*}
Obviously, $|\C_1|=Mh$.
It is easy to see that 
\begin{equation}\label{eq:M-channel_construction_C1_D}
D_{\widehat{S}_m}(i,j) = \begin{cases}
    \{0\}\times d^*(S) & \text{if } i=j=m, \\
    \emptyset & \text{otherwise}.
\end{cases}
\end{equation}
As $\C'$ is a CAC, it holds that $D_{\widehat{S}_m}$, $\widehat{S}_m\in\C_1$, are mutually disjoint.

Secondly, let $\Omega=\{(1,g)\in\Z_{2t-1}\times\Z_{L'}:\,g\in\Z_{L'}\}\ (=\{1\}\times\Z_{L'})$.
For $a\in\Omega$, define a $w$-subset in $\mathcal{I}_M\times\Z_L$ by
\begin{equation}\label{eq:M-channel_construction_C2}
S_a \triangleq \bigcup_{m=1}^{M}\{(m,ka)\in\mathcal{I}_M\times\Z_L:\,k=(m-1)t,(m-1)t+1,\ldots,mt-1\}.
\end{equation}
Let 
\begin{align*}
\C_2=\bigcup_{a\in\Omega}S_a.
\end{align*}
Clearly, $|\C_2|=L'$.
It is not difficult to see that
\begin{equation}\label{eq:M-channel_construction_C2_D}
D_{S_a}(i,j) = \begin{cases}
    \{ka\in\Z_L:\,k=\pm1,\pm2,\ldots,\pm(t-1)\} & \text{if } i=j, \\
    \{ka\in\Z_L:\,(i-j-1)t+1\leq k\leq (i-j+1)t-1\} & \text{if }i\neq j.
\end{cases}
\end{equation}
For convenience, let
\begin{align*}
A_{i,j}=\{x\in\Z:\,(i-j-1)t+1\leq x\leq (i-j+1)t-1\}.
\end{align*}
We will claim that $D_{S_a}\cap D_{S_b}=\emptyset$ for any two distinct $a,b\in\Omega$.
Suppose to the contrary that $D_{S_a}\cap D_{S_b}\neq\emptyset$, where $a=(1,g), b=(1,h)$ for some distinct $g,h\in\Z_{L'}$.
There are two cases according to~\eqref{eq:M-channel_construction_C2}.

\textit{Case 1.} 
If $D_{S_a}(i,i)\cap D_{S_b}(i,i)\neq\emptyset$ for some $i$, then there exist $k,k'\in\{\pm1,\pm2,\ldots,\pm(t-1)\}$ such that $k(1,g)=k'(1,h)$ in $\Z_{2t-1}\times\Z_{L'}$.
This implies that $k=k'$ and thus $k(g-h)=0$ (mod $L'$).
Since all prime factors of $L'$ are larger than or equal to $2w-1$, which is strictly larger than $2t-1$, we have $\gcd(k,L')=1$.
This further implies that $g=h$, which is a contradiction.

\textit{Case 2.}
If $D_{S_a}(i,j)\cap D_{S_b}(i,j)\neq\emptyset$ for some $i\neq j$, then there exist integers $k,k'\in A_{i,j}$ such that $k(1,g)=k'(1,h)$ in $\Z_{2t-1}\times\Z_{L'}$.
Observe that $A_{i,j}$ is consist of consecutive $2t-1$ integers.
Then, we get $k=k'$, and thus $k(g-h)=0$ (mod $L'$).
Moreover, since $i,j$ are distinct elements in $\mathcal{I}_M$, we have $-w+1\leq k\leq w+1$ and $k\neq 0$.
Again, since all prime factors of $L'$ are larger than or equal to $2w-1$, one has $\gcd(k,L')=1$.
Hence, it follows that $g=h$, and a contradiction occurs.

Furthermore, by~\eqref{eq:M-channel_construction_C1_D} and \eqref{eq:M-channel_construction_C2_D}, it is obvious that $D_{S}\cap D_{S'}=\emptyset$ whenever $S\in\C_1$ and $S'\in\C_2$.
We conclude that $\C=\C_1\uplus\C_2$ is a code in $\MCCAC(M,L,w)$.
\end{proof}

\begin{example}\label{ex:Mchannel_3_37}\rm
Let $w=6,M=3$, and $L'=37$. We have $L=3\cdot 37$.
One can check that $\{0,1,2,3,4,5\}$ and $\{0,6,12,18,24,30\}$ form a code $\C'\in\CAC(37,6)$. 
By Theorem~\ref{thm:Mchannel_construction}, we have a code in $\MCCAC(3,111,6)$ with $2\cdot 3+37=43$ codewords.
By the CRT correspondence $\theta:\Z_{111}\to\Z_3\times\Z_{37}$ given in~\eqref{eq:CRT-correspondence}, the first $6$ codewords are $\{(m,0),(m,75),(m,39),(m,3),(m,78)$, $(m,42)\}$, $\{(m,0),(m,6),(m,12),(m,18),(m,24),(m,30)\}$ for $m=1,2,3$.
The remaining $37$ codewords are $\{(1,0),(1,a),(2,2a),(2,3a),(3,4a),(3,5a)\}$, where $a\in\Omega=\{(1,g)\in\Z_3\times\Z_{37}:\,g\in\Z_{37}\}$.
The set $\Omega$ consists exactly of the multiples of 3 in $\Z_{111}$.
\end{example}

Note that the base CACs for the construction in Theorem~\ref{thm:Mchannel_construction} are not necessarily equi-difference nor tight.

\begin{corollary}\label{coro:Mchannel_construction}
If the based CAC described in Theorem~\ref{thm:Mchannel_construction} has $(L'-1)/(2w-2)$ codewords, then there exists a code $\C\in \MCCAC(M, L, w)$ with
\begin{equation}\label{eq:Mchannel_lower}
|\C|=\frac{M(M-1)L}{(2w-M)(w-1)}+\frac{M(L-1)}{2w-2}.
\end{equation}
\end{corollary}
\begin{proof}
The result is obtained by plugging $L'=ML/(2w-M)$ and $h=(L'-1)/(2w-2)$ into $|\C|=Mh+L'$.
\end{proof}

The code size in~\eqref{eq:Mchannel_lower} coincides with the upper bound given in Theorem~\ref{thm:UB-w>M}, namely~\eqref{eq:K-UB-w>M}.
However, the code constructed in Corollary~\ref{coro:Mchannel_construction} is not optimal since its length contains a factor, say $2\frac{w}{M}-1$, that is smaller than $2w-1$.
Therefore, Theorem~\ref{thm:UB-w>M} is not applicable in this case. 
In the following, we derive another upper bound for CACs of this length.

For any positive integer $n$, let $\tau(n)$, called the \emph{number-of-divisor function}, denote the number of positive divisors of $n$.

\begin{theorem}\label{thm:Mchannel_bounds}
Let $M,L,L'$ and $w$ be positive integers with $w>M\geq 3$ such that $M$ divides $w$, $L=(2\frac{w}{M}-1)L'$, and all prime factors of $L'$ are larger than or equal to $2w-1$. 
Then, 
\begin{align*}
K(M,L,w) \leq \frac{M(M-1)(L+(\tau(2\frac{w}{M}-1)-1)(L-L'))}{(2w-M)(w-1)}+\frac{M(L-1)+2(w-M)}{2w-2}.
\end{align*}
\end{theorem}
\begin{proof}
Let $\C$ be any code in $\MCCAC(M,L,w)$.
We adopt the notation from the proof of Theorem~\ref{thm:UB-w>M} for convenience.
Let $\C_1=\{S\in\C:\,e_S=1\}$, $\C_2=\{S\in\C:\,e_S=M\}$, and $\C_3=\{S\in\C:\,2\leq e_S\leq M-1\}$.

For $i\in\mathcal{I}_M$, let 
\begin{align*}
\C^{(i)}\triangleq\{S_i:\,S\in\C\}
\end{align*}
and
\begin{align*}
\mathcal{E}^{(i)}\triangleq\{A\in\C^{(i)}:\,A\text{ is exceptional}\}
\end{align*}
denote the collections of all subsets and all exceptional subsets associated with the $i$-th channel, respectively.
For $A\in\mathcal{E}^{(i)}$, let $H_A=\HH(A-A)$ for notational convenience.

For any fixed index $i\in\mathcal{I}_M$, by the same arguments used to obtain~\eqref{eq:2channel_w-1pr_optimal_4}, we derive
\begin{align*}
\sum_{A\in\mathcal{E}^{(i)}}|d^*(A)| \geq \left(\sum_{A\in\mathcal{E}^{(i)}}2|A|-2\right) - \left(2\frac{w}{M}-2\right).
\end{align*}
Note that $|d^*(A)|\geq 2|A|-2$ for all $A\in\C^{(i)}\setminus\mathcal{E}^{(i)}$.
Therefore, 
\begin{equation}\label{eq:Mchannel_bounds_1}
\sum_{A\in\C^{(i)}}|d^*(A)| \geq \left(\sum_{A\in\C^{(i)}}2|A|-2\right) - 2\left(\frac{w}{M}-1\right).
\end{equation}
Summing~\eqref{eq:Mchannel_bounds_1} over all $i\in\mathcal{I}_M$, by the disjoint-difference-array property, yields
\begin{equation}\label{eq:Mchannel_bounds_2}
M(L-1) \geq 2(w-1)|\C_1| + 2(w-M)|\C_2| + 2\sum_{S\in\C_3}(w-e_S) - 2M\left(\frac{w}{M}-1\right).
\end{equation}


Now, we consider $D_S(i,j)$ ($=S_i-S_j$) for distinct indices $i,j\in\mathcal{I}_M$.
Note that $S_i-S_j=\emptyset$ for $S\in\C_1$.
For $S\in\C_2\cup\C_3$, denote by $H_S^{(i,j)}=\HH(S_i-S_j)$ for notational convenience.
By plugging $A=S_i$ and $B=-S_j$ into \eqref{eq:Kneser2}, we have
\begin{equation}\label{eq:Mchannel_bounds_3}
|S_i-S_j| \geq |S_i|+|S_j|-|H_S^{(i,j)}|.
\end{equation}

For two distinct channel indices $i,j\in\mathcal{I}_M$, let
\begin{align*}
\mathcal{E}^{(i,j)} \triangleq \{S\in\C_2\cup\C_3:\,\{S_i,S_j\} \text{ is an exceptional pair}\}
\end{align*}
denote the collections of all exceptional pairs associated with the $i$- and $j$-th channels.
Note that $|S_i-S_j|\geq |S_i|+|S_j|-1$ when $\{S_i,S_j\}$ is not an exceptional pair.
By~\eqref{eq:Mchannel_bounds_3} and the disjoint-difference-array property, we have
\begin{align*}
L &= |\Z_L| \geq \sum_{S\in\C_2\cup\C_3} |S_i-S_j| \\
&\geq \sum_{S\in\C_2\cup\C_3} \left(|S_i|+|S_j|-1\right) - \sum_{S\in\mathcal{E}^{(i,j)}}\left(|H_S^{(i,j)}|-1\right).
\end{align*}
By the same arguments used to obtain~\eqref{eq:UB-w>M-2}, we can sum the above inequality over all distinct indices $i,j\in\mathcal{I}_M$ to derive
\begin{equation}\label{eq:Mchannel_bounds_4}
\begin{split}
M(M-1)L \geq & \ (M-1)(2w-M)|\C_2| + \sum_{S\in\C_3}(e_S-1)(2w-e_S) \\
&- \sum_{i,j\in\mathcal{I}_M,i\neq j}\sum_{S\in\mathcal{E}^{(i,j)}}\left(|H_S^{(i,j)}|-1\right).
\end{split}
\end{equation}

By multiplying the inequality in~\eqref{eq:Mchannel_bounds_2} by $(2w-M)/2$ and combining the result with~\eqref{eq:Mchannel_bounds_4}, we get
\begin{align*}
& (w-1)(2w-M)(|\C_1|+|\C_2|) + \sum_{S\in\C_3}(e_S-1)(2w-e_S) \\
\leq & \ M(M-1)L + \sum_{i,j\in\mathcal{I}_M,i\neq j}\sum_{S\in\mathcal{E}^{(i,j)}}\left(|H_S^{(i,j)}|-1\right) + \frac{M(L-1)(2w-M)}{2} + (2w-M)(w-M)
\end{align*}
Note that we have dealt with the summation $\sum_{S\in\C_3}(e_S-1)(2w-e_S)$ in the proof of Theorem~\ref{thm:UB-w>M}.
By~\eqref{eq:UB-w>M-5}, the above inequality can be simplified to
\begin{equation}\label{eq:Mchannel_bounds_5}
|\C| \leq \frac{M(M-1)L+\sum_{i,j\in\mathcal{I}_M,i\neq j}\sum_{S\in\mathcal{E}^{(i,j)}}\left(|H_S^{(i,j)}|-1\right)}{(2w-M)(w-1)}+\frac{M(L-1)+2(w-M)}{2w-2}.
\end{equation}

For $S\in\mathcal{E}^{(i,j)}$, by Lemma~\ref{lem:exceptional_new}(i), $|H_S^{(i,j)}|\geq 2$.
It follows from Proposition~\ref{prop:stabilizer}(iii) that $|H_S^{(i,j)}|$ divides $|S_i-S_j|$.
Note that $|H_S^{(i,j)}|$ divides $L$ since $H_S^{(i,j)}$ is a subgroup of $\Z_L$.
For a given subset $H$ of $\Z_L$, define $\delta(H)$ as the collection of elements $S\in\mathcal{E}^{(i,j)}$ such that $H_S^{(i,j)}=H$.
Since, for two distinct $S,S'\in\delta(H)$, $(S_i-S_j)\cap(S'_i-S'_j)=\emptyset$ due to the disjoint-difference-array property, we have
\begin{align*}
    \sum_{S\in\delta(H)}|S_i-S_j| \leq L.
\end{align*}
This implies that
\begin{align*}
    |\delta(H)| \leq \frac{L}{|H|}.
\end{align*}
In other words, there are at most $L/|H|$ exceptional pairs $\{S_i,S_j\}$ over all $\mathcal{E}^{(i,j)}$ such that their set of stabilizers is $H$.
By Lemma~\ref{lem:exceptional_new}(i) again, $|H_S^{(i,j)}|\leq |S_i|+|S_j|-2 \leq w-2$, which is strictly less than any prime factor of $L'$.
By Lagrange's theorem, the value $|H_S^{(i,j)}|$ divides $(2\frac{w}{M}-1)$.
Moreover, $|H_S^{(i,j)}|\geq 2$.
That means, there are at most $\tau(2\frac{w}{M}-1)-1$ candidates for the set of stabilizers of an exceptional pair associated with $\mathcal{E}^{(i,j)}$.
Hence, we conclude that
\begin{align*}
\sum_{S\in\mathcal{E}^{(i,j)}}\left(|H_S^{(i,j)}|-1\right) &\leq \left(\tau(2\frac{w}{M}-1)-1\right) \cdot \frac{L}{|H|}\left(|H|-1\right) \\
&\leq \left(\tau(2\frac{w}{M}-1)-1\right)(L-L'),
\end{align*}
where the second inequality is due to $|H|$ being a factor of $(2\frac{w}{M}-1)$.
Finally, we get
\begin{align*}
\sum_{i,j\in\mathcal{I}_M,i\neq j}\sum_{S\in\mathcal{E}^{(i,j)}}\left(|H_S^{(i,j)}|-1\right) \leq M(M-1) \left(\tau(2\frac{w}{M}-1)-1\right)(L-L').
\end{align*}
The result follows by~\eqref{eq:Mchannel_bounds_5}.
\end{proof}

\section{Two Extensions}\label{sec:extension}

\subsection{At Most One-Packet Per Time Slot MC-CACs}\label{sec:AM-OPPTS}

A code $\C\in\MCCAC(M,L,w)$ is said to have \textit{at most one-packet per time slot (AM-OPPTS)} if, for all codeword $S\in\C$, the element ``one'' appears at most once in each column~\cite{WFLWG23}.
That is, $|\{m:\,(m,t)\in S\}|\leq 1$ for each $t\in\Z_L$.
Obviously, the AM-OPPTS property is meaningless if $L<w$.
Let AM-OPPTS MC-CAC$(M,L,w)$ be the collection of codes in $\MCCAC(M,L,w)$ with AM-OPPTS property.
The maximum size of a code in AM-OPPTS $\MCCAC(M,L,w)$ is denoted by $A(M,L,w)$, i.e.,
\begin{align*}
A(M,L,w)\triangleq\max\{|\C|:\,\C\in\text{AM-OPPTS }\MCCAC(M,L,w)\}.
\end{align*}
By definitions, we have $A(M,L,w)\leq K(M,L,w)$ for any integers $M,L$ and $w$.

The following upper bound of $A(M,L,w)$ is due to \cite[Corollaries 3.6, 3.7]{WFLWG23}.

\begin{theorem}[\!\!\cite{WFLWG23}]\label{thm:UB-A-any-w}
Let $M,L$ and $w$ be positive integers. 
If the prime factors of $L$ are all larger than or equal to $2w-1$, we have
\begin{equation}\label{eq:A-UB-any-w}
A(M,L,w)\leq \frac{M(M-1)(L-1)}{w(w-1)}+\frac{M(L-1)}{2w-2}.
\end{equation}
\end{theorem}

Note that Theorem~\ref{thm:UB-A-any-w} can be derived by the same argument we used in the proof of Theorem~\ref{thm:UB-any-w}, which is provided in Appendix~\ref{Appendix:A}.

The upper bounds in Theorem~\ref{thm:UB-A-any-w} can be further improved in the case where $M<w$, which is shown in the following theorem.

\begin{theorem}\label{thm:UB-A-w>M}
Let $M,L$ and $w$ be positive integers with $M<w$. 
If the prime factors of $L$ are all larger than or equal to $2w-1$, we have
\begin{equation}\label{eq:A-UB-w>M}
A(M,L,w)\leq \frac{M(M-1)(L-1)}{(2w-M)(w-1)}+\frac{M(L-1)}{2w-2}.
\end{equation}
\end{theorem}
\begin{proof}
The characteristic of a code $\C$ in AM-OPPTS $\MCCAC(M,L,w)$ is that, $0\notin D_S(i,j)$ for any codeword $S$ involved and any distinct indices $i,j\in\mathcal{I}_M$.
Therefore, the inequality in~\eqref{eq:A-UB-w>M} can be obtained in the same arguments given in the proof of Theorem~\ref{thm:UB-w>M}, but just modifying~\eqref{eq:UB-w>M-Ds} to be 
\begin{align*}
\sum_{S\in\C_2\cup\C_3}|D_S(i,j)|\leq |\Z^*_L| = L-1,
\end{align*}
which changes the last term $M(M-1)L$ on the right-hand-side in~\eqref{eq:UB-w>M-3} to $M(M-1)(L-1)$.
Hence, the result follows.
\end{proof}

Based on the constructions developed in Sections~\ref{sec:MC-CAC_2} and \ref{sec:MC-CAC_M}, we obtain the following results.

\begin{theorem}\label{thm:2channel_AMOPPTS}
Let $p_1<\cdots<p_n$ be primes such that $p_i-1$ is divisible by $2w-2$ for each $i$.
For each $1\leq i\leq n$, if there is a code in $\CACe(p_i,w)$ with $(p_i-1)/(2w-2)$ codewords and the two conditions in~\eqref{eq:Q1} and\eqref{eq:Q2} hold, then for any positive integers $r_1,\ldots,r_n$, 
\begin{align*}
\prod_{i=1}^{n}p_i^{r_i} + \frac{\left(\prod_{i=1}^{n}p_i^{r_i}\right)-1}{w-1}-1 \leq A(2,L,w) \leq  \prod_{i=1}^{n}p_i^{r_i} + \frac{\left(\prod_{i=1}^{n}p_i^{r_i}\right)-1}{w-1},
\end{align*}
where $L=(w-1)p_1^{r_i}\cdots p_n^{r_n}$.
\end{theorem}
\begin{proof}
The lower bound can be derived from the construction in the proof of Theorem~\ref{thm:2channel_w-1pr_construction} by ignoring the last codeword, say the unique one in $\C_3$.
The upper bound is a consequence of Theorem~\ref{thm:2channel_w-1pr_optimal} because $A(2,L,w)\leq K(2,L,w)$.
\end{proof}

\begin{theorem}\label{thm:Mchannel_AMOPPTS}
Let $M,L,L'$, and $w$ be positive integers such that $w>M\geq 3$, $M$ divides $w$, $L=(2\frac{w}{M}-1)L'$, and all prime factors of $L'$ are larger than or equal to $2w-1$. 
If there is a code in $\CACe(L',w)$ with $(L'-1)/(2w-2)$ codewords, then 
\begin{align*}
\frac{M(M-1)L}{(2w-M)(w-1)}+&\frac{M(L-1)}{2w-2} \leq A(M,L,w) \\ 
&\leq \frac{M(M-1)(L+(\tau(2\frac{w}{M}-1)-1)(L-L'))}{(2w-M)(w-1)}+\frac{M(L-1)+2(w-M)}{2w-2}.
\end{align*}
\end{theorem}
\begin{proof}
In the proof of Theorem~\ref{thm:Mchannel_construction}, it is easy to see from~\eqref{eq:M-channel_construction_C2_D} that $0\notin D_{S_a}(i,j)$ for $i\neq j$.
This implies that the code obtained in Theorem~\ref{thm:Mchannel_construction} is has the AM-OPPTS property.
The upper bound is a consequence of Theorem~\ref{thm:Mchannel_bounds} because $A(M,L,w)\leq K(M,L,w)$.
\end{proof}

\subsection{Mixed-Weight Codes}\label{sec:mixed-weight}

Mixed-weight CACs were introduced~\cite{LWXZ25} for the purpose of increasing the throughput and reducing the access delay of some potential users with higher priority.
This subsection will provide some constructions of mixed-weight CACs, which can be seen as generalizations of Theorems~\ref{thm:single_w-1pr_construction} and \ref{thm:single_pr_construction}.
Furthermore, the concept of mixed-weight will also be extended to a multichannel version by relaxing the identical-weight constraint in prior studies of MC-CACs.


\begin{definition}\label{def:CAC-mix}\rm
Let $L$ be a positive integer and $\mathcal{W}$ be a set of positive integers.
A \textit{mixed-weight CAC} $\C$ of length $L$ with \textit{weight-set} $\mathcal{W}$ is a collection of subsets of $\Z_L$ such that (i) each subset is of size in $\mathcal{W}$; and (ii) $\C$ satisfies the disjoint-difference-set property as shown in~\eqref{eq:CAC}.
\end{definition}

Let $\CAC(L,\mathcal{W})$ denote the class of all mixed-weight CACs of length $L$ with weight-set $\mathcal{W}$.

Under the assumption that a code in $\CAC(p_{i}^{r_i},w_{i}^*)$ exists for each $i=1,2,\ldots,n$, we will first construct a mixed-weight CAC of length $p_1^{r_1}\cdots p_n^{r_n}$ with weight-set and then construct a mixed weight CAC of length $(w-1)p_1^{r_1}\cdots p_n^{r_n}$ with weight-set 
$\mathcal{W}=\{w-1, w_{1}^*,w_{2}^*,\ldots,w_{n}^*\}$.
Note that the two constructions are obtained by modifying those in Theorems~\ref{thm:single_pr_construction} and \ref{thm:single_w-1pr_construction}, respectively.

\begin{theorem}\label{thm:single_pr_mix}
Let $w_i,r_i$, $1\leq i\leq n$, be positive integers.
Assume that $p_1,p_2,\ldots,p_n$ are primes such that $p_i\geq 2w_i-1$ for each $i$.
Suppose, for each $1\leq i\leq n$, there exists a code $\mathcal{A}_i\in\CAC(p_i^{r_i},w_i)$ that contains $m_i$ equi-difference codewords.
Then there exists a code $\C\in\CAC(L,\mathcal{W})$ 
with $$|\C|=\sum_{i=1}^{n}m_i\prod_{j=i+1}^{n}p_j^{r_j},$$ where $L=p_1^{r_1}\cdots p_n^{r_n}$ and $\mathcal{W}=\{w_i\}_{i=1}^{n}$, that contains $m_i\prod_{j=i+1}^{n}p_j^{r_j}$ codewords of weight $w_i$ for each $i$.
\end{theorem}

\begin{proof}
Let $\Gamma_i$ be the set of $m_i$ generators of the equi-difference codewords of $\mathcal{A}_i$.
The proof can be done with the same arguments as in the proof of Theorem~\ref{thm:single_pr_construction} by defining
\begin{align*}
\widehat{\Gamma}_i\triangleq\{(0,\ldots,0,g,x_{i+1},\ldots,x_n) &\in \Z_{p_1^{r_1}}\times \cdots \times \Z_{p_n^{r_n}}:\, \\ 
& g\in\Gamma_i, x_j\in\Z_{p_j^{r_j} }\text{ for }i<j\leq n\}
\end{align*} 
and replacing the definition of $S_a$, for $a\in\widehat{\Gamma}_i$, in~\eqref{eq:sginel_pr_Sa} with
\begin{align*}
S_a=\{ja\in\Z_{p_1^{r_1}}\times \cdots \times \Z_{p_n^{r_n}} :\, j=0,1,2,\ldots,w_i-1\}.
\end{align*}
The rest of the proof is omitted.
\end{proof}

\begin{theorem}\label{thm:single_w-1pr_mix}
Let $w_i,r_i$, $1\leq i\leq n$, and $w$ be positive integers.
Assume that $p_1,p_2,\ldots,p_n$ are primes such that $p_i\geq 2w-1,2w_i-1$ for each $i$.
Suppose, for each $1\leq i\leq n$, the two conditions given in~\eqref{eq:Q1} and \eqref{eq:Q2} hold and there exists a code $\mathcal{A}_i\in\CAC(p_i^{r_i},w_i)$ that contains $m_i$ equi-difference codewords.
Then there exists a code $\C\in\CAC(L,\mathcal{W})$ with 
\begin{align*}
|\C|=\frac{p_1^{r_1}\cdots p_n^{r_n}+1}{2}+\sum_{i=1}^{n}m_i\prod_{j=i+1}^{n}p_j^{r_j},    
\end{align*}
where $L=(w-1)p_1^{r_1}\cdots p_n^{r_n}$ and $\mathcal{W}=\{w-1\}\cup\{w_i\}_{i=1}^{n}$. 
In particular, $\C$ contains $(p_1^{r_1}\cdots p_n^{r_n}+1)/2$ codewords of weight $w-1$ and $m_i\prod_{j=i+1}^{n}p_j^{r_j}$ codewords of weight $w_i$ for each $i$.
\end{theorem}
\begin{proof}
We adopt the notation and definitions used in the proof of Theorem~\ref{thm:single_w-1pr_construction}.
Let $L'=p_1^{r_1}\cdots p_n^{r_n}$ for convenience.
Note that $\Z_L \cong \Z_{w-1}\times\Z_{L'}$ and $\Z_{L'} \cong \Z_{p_1^{r_1}}\times\cdots\times\Z_{p_n^{r_n}}$.
We aim to construct $n+1$ classes of codewords, say $\C_{w-1},\C_{w_1},\ldots,\C_{w_n}$, consisting of codewords with weights $w-1, w_1,\ldots,w_n$, respectively.

Firstly, for $a\in\widehat{Q}$, define 
\begin{align*}
S_a=\{j(1,a)\in\Z_{w-1}\times\Z_{L'}:\,j=0,1,2\ldots,w-2\}
\end{align*}
be the $(w-1)$-subset generated by $a$, where $\widehat{Q}$ is defined in~\eqref{eq:single_w-1pr_generator}.
It follows from the proof of Theorem~\ref{thm:single_w-1pr_construction} that $d^*(S_a)$, $a\in\widehat{Q}$, are mutually disjoint.
Note that only \textit{Case 2} in the proof of Theorem~\ref{thm:single_w-1pr_construction} is needed here.
Moreover, we have 
\begin{equation}\label{eq:single_w-1pr_mix_class_1}
\bigcup_{a\in\widehat{Q}} d^*(S_a)\subseteq \Z_{w-1}^*\times\Z_{L'}^*.
\end{equation}
Then, define the codeword $T=\{(j,0)\in\Z_{w-1}\times\Z_{L'}:\, j=0,1, \cdots w-2\}$.
Obviously, $d^*(T)=\Z_{w-1}^*\times\{0\}$, and it is disjoint from $d^*(S_a)$ for every $a\in\widehat{Q}$.
Let $\C_{w-1}=\{T\}\cup\{S_a:\,a\in\widehat{Q}\}$.
It follows from~\eqref{eq:single_w-1pr_mix_class_1} that
\begin{equation}\label{eq:single_w-1pr_mix_class_2}
\bigcup_{S\in\C_{w-1}} d^*(S)\subseteq \Z_{w-1}^*\times\Z_{L'}.
\end{equation}
By~\eqref{eq:single_w-1pr_generator}, we therefore have $|\C_{w-1}|=1+(L'-1)/2$.

Finally, let $\C_{w_i}$ be the sets of codewords of weight $w_i$ that are obtained from $\mathcal{A}_i$ by the construction given in Theorem~\ref{thm:single_pr_mix} and raised from subsets of $\Z_{L'}$ to $\{0\}\times\Z_{L'}$.
More precisely, each codeword in $\C_{w_i}$ is of the form $\{(0,s)\in\Z_{w-1}\times\Z_{L'}:\,s\in S\}$, where $S$ is a codeword obtained from $\mathcal{A}_i$ by the construction given in Theorem~\ref{thm:single_pr_mix}.
It is clear that 
\begin{align*}
\bigcup_{S\in\C_{w_1}\cup\cdots\cup\C_{w_n}} d^*(S) \subseteq \{0\}\times\Z_{L'}.
\end{align*}
Hence, it follows from~\eqref{eq:single_w-1pr_mix_class_2} and Theorem~\ref{thm:single_pr_mix} that $\C_{w-1}\cup\C_{w_1}\cup\cdots\cup\C_{w_n}$ forms the desired mixed-weight code.
\end{proof}

\begin{example}\rm
Let $w=4$, $n=2$, $r_1=r_2=1$ and $p_1=23, p_2=47$.
We have $L=3243$.
Consider $w_1=9,w_2=6$ and $m_1=1,m_2=2$, and assume that the set of generators of equi-difference codewords in $\mathcal{A}_1$ and $\mathcal{A}_2$ are $\Gamma_1=\{1\}$ and $\Gamma_2=\{1,7\}$, respectively.

By Theorems~\ref{thm:single_w-1pr_construction} and \ref{thm:single_w-1pr_mix}, we have $\widehat{Q}=\widehat{Q}_1\cup\widehat{Q}_2$, where $\widehat{Q}_1=\{(g,x)\in\Z_{23}\times\Z_{47}:\,g\in Q_1, x\in\Z_{47}\}$, $\widehat{Q}_2=\{(0,g)\in\Z_{23}\times\Z_{47}:\,g\in Q_2\}$, $Q_1=\{1,2,3,4,6,8,9,12,13,16,18\}$, and $Q_2=\{1,2,3,4,6,7,8,9,12,14,16,17,18,21,24,25,27,28, 32,34,36,37,42\}$.
Let $\C_3=\{T\}\cup\{S_a:\,a\in\widehat{Q}\}$, where $T=\{(0,0),(1,0),(2,0)\}$ and $S_a=\{(0,0),(1,a),(2,2a)\}$.
Note that the elements in those codewords are considered in $\Z_3\times\Z_{1081}$, and can be mapped to elements in $\Z_{3243}$ by the CRT correspondence~\eqref{eq:CRT-correspondence}.
For example, $T=\{0,1081,2162\}$ and $S_{(1,0)}=\{0,2209,1175\}$.

Following the construction in Theorems~\ref{thm:single_pr_mix} and \ref{thm:single_w-1pr_mix}, we have $\C_9=\{S_{(1,x)}:\,x=0,1,\ldots,46\}$ and $\C_6=\{S_{(0,1)},S_{(0,7)}\}$, where
\begin{align*}
S_{(1,x)} = \{(0,j,jx)\in\Z_3\times\Z_{23}\times\Z_{47}:\, j=0,1,\ldots,8\}
\end{align*}
and
\begin{align*}
S_{(0,k)} = \{(0,0,jk)\in\Z_3\times\Z_{23}\times\Z_{47}:\, j=0,1,\ldots,5\}
\end{align*}
for $k=1,7$.
Similarly, by the CRT correspondence $\Z_3\times\Z_{23}\times\Z_{47}\cong\Z_{3243}$, we have $S_{(0,1)} =\allowbreak \{0,897,1035,1932,2070,310\}$, $S_{(0,7)}=\{0,552,759,1518,2277,3036\}$, and $S_{(1,0)}=\{0,141, \allowbreak 282,1128,1269,1410,2256,2397,2538\}$ for instance.
Hence we have $|\C_3|=541$ , $|\C_9|=47$, $|\C_6|=2$.
\end{example}

We now extend the concept of mixed-weight to MC-CACs.

\begin{definition}\label{def:MCCAC_mix}\rm
Let $M,L$ be two positive integers and $\mathcal{W}$ be a set of positive integers.
A \textit{mixed-weight MC-CAC} $\C$ of length $L$ with \textit{weight-set} $\mathcal{W}$ in $M$ channels is a collection of subsets of $\mathcal{I}_M\times\Z_L$ such that (i) each subset is of size in $\mathcal{W}$; and (ii) $\C$ satisfies the disjoint-difference-array property as shown in~\eqref{eq:MCCAC}.
\end{definition}

Let $\MCCAC(M,L,\mathcal{W})$ denote the class of all mixed-weight MC-CACs of length $L$ with weight-set $\mathcal{W}$ for $M$ channels.


The construction of MC-CACs in Theorem~\ref{thm:2channel_w-1pr_construction} relies on the CAC constructions in Theorems~\ref{thm:single_w-1pr_construction} and \ref{thm:single_pr_construction}. 
Since these admit mixed-weight generalizations (Theorems~\ref{thm:single_w-1pr_mix} and \ref{thm:single_pr_mix}, respectively), Theorem~\ref{thm:2channel_w-1pr_construction} extends naturally to the mixed-weight case.
 
\begin{theorem}\label{thm:2channel_mix}
Let $w_i,r_i$, $1\leq i\leq n$, and $w$ be positive integers.
Assume that $p_1,p_2,\ldots,p_n$ are primes such that $p_i\geq 2w-1,2w_i-1$ for each $i$.
Suppose, for each $1\leq i\leq n$, the two conditions given in~\eqref{eq:Q1} and \eqref{eq:Q2} hold, and there exists a code in $\CAC(p_i^{r_i},w_i)$ that contains $m_i$ equi-difference codewords.
Then there exists a mixed-weight code $\C\in\CAC(2,L,\mathcal{W})$ with 
\begin{align*}
|\C| = 1+ \prod_{i=1}^{n}p_i^{r_i} + \sum_{i=1}^{n}\frac{2m_i(p_i^{r_i}-1)}{p_i-1}\prod_{j=i+1}^{n}p_j^{r_j}
\end{align*}
where $L=(w-1)p_1^{r_1}\cdots p_n^{r_n}$ and $\mathcal{W}=\{w,w-1\}\cup\{w_i\}_{i=1}^{n}$. 
In particular, $\C$ contains $\prod_{i=1}^{n}p_i^{r_i}$ codewords of weight $w$, one codeword of weight $w-1$, and $\big(2m_i(p_i^{r_i}-1)/(p_i-1)\big)\prod_{j=i+1}^{n}p_j^{r_j}$ codewords of weight $w_i$, for $i=1,2,\ldots,n$.
\end{theorem}
\begin{proof}
The proof follows from the combination of the proofs of Theorems~\ref{thm:2channel_w-1pr_construction}
and \ref{thm:single_pr_mix}.
The codeword of weight $w-1$ is given by $\{(2,(j,0)):\,j=0,1,\ldots,w-2\}$, which is a variation of $U$, the last codeword in the proof of Theorems~\ref{thm:2channel_w-1pr_construction}.
\end{proof}

We also have the following mixed-weight version of Theorem~\ref{thm:Mchannel_construction}.

\begin{theorem}\label{thm:Mchannel_mix}
Let $M,L,L'$, and $w$ be positive integers such that $w>M\geq 3$, $M$ divides $w$, $L=(2\frac{w}{M}-1)L'$, and all prime factors of $L'$ are larger than or equal to $2w-1$. 
If, for $1\leq m\leq M$, there is a code $\C_m\in\CAC(L',\mathcal{W}_m)$ with $|\C_m|=h_m$ codewords, then there exists a code $\C\in \MCCAC(M, L, \mathcal{W})$ with $|\C|=L'+\sum_{m=1}^{M}h_m$ codewords, where $\mathcal{W}=\{w\}\cup\bigcup_{m=1}^{M}\mathcal{W}_m$ and each $\mathcal{W}$ is an arbitrary set of positive integers.
In particular, $\C$ contains $L'$ codewords of weight $w$, and the number of codewords of weight $w'$ in $\C$ equals that in $\bigcup_{m=1}^{M}\C_m$.
\end{theorem}
\begin{proof}
The proof can be obtained by a straightforward adaptation of the proof of Theorem~\ref{thm:Mchannel_construction}. 
Specifically, the codewords on channel $m$ are constructed using $\C_m$.
That is, for $S\in\C_m$, $1\leq m\leq M$, let $\widehat{S} = \{(m,(0,a))\in\mathcal{I}_M\times\Z_L:\,a\in S\}$.
\end{proof}

\section{Concluding Remarks}\label{sec:conclusion}

This paper investigates MC-CACs from the perspectives of both constructions and theoretical upper bounds on the maximum code size.
We first introduce the notion of exceptional codewords for MC-CACs and establish a series of structural properties using techniques from additive combinatorics, including Kneser’s theorem.
These results yield new upper bounds on the maximum code size of MC-CACs, especially for the case where the channel number is less than the maximum number of active users.
Next, we generalize two previously known constructions of (single-channel) CACs in Theorems~\ref{thm:single_w-1pr_construction} and \ref{thm:single_pr_construction}, which in turn lead to a construction of MC-CACs with $M=2$ channels, as presented in Theorem~\ref{thm:2channel_w-1pr_construction}. 
Necessary conditions for these constructions to be optimal are characterized in Theorems~\ref{thm:single_w-1pr_optimal}, \ref{thm:single_pr_optimal}, and \ref{thm:2channel_w-1pr_optimal}, respectively.
A construction of MC-CACs for $M\geq 3$ such that $M$ divides $w$ is given in Theorem~\ref{thm:Mchannel_construction}, and a corresponding upper bound on the maximum code size is established in Theorem~\ref{thm:Mchannel_bounds}.
Our results on MC-CACs can be further extended to AM-OPPTS MC-CACs in Theorems~\ref{thm:UB-A-w>M} -- \ref{thm:Mchannel_AMOPPTS}.
Finally, we study mixed-weight MC-CACs for the first time and propose a series of constructions in Theorems~\ref{thm:2channel_mix} and \ref{thm:Mchannel_mix}, based on the structures of mixed-weight CACs given in Theorems~\ref{thm:single_w-1pr_mix} and \ref{thm:single_pr_mix}.

It has been shown in~\cite[Theorem 1]{SW10} that 
\begin{equation}\label{eq:asym_M=1}
\limsup_{L\to\infty} \frac{K(L,w)}{L} = \frac{1}{2w-2}.
\end{equation}
When considering multichannel CACs, if $M=2$, it reveals from Theorem~\ref{thm:2channel_w-1pr_optimal} that 
\begin{equation}\label{eq:asym_M=2}
\frac{K(2,L,w)}{L} \sim \frac{1}{w-1}+\frac{1}{(w-1)^2} \quad  \text{as} \quad  L\to\infty
\end{equation}
under some specific parameter settings.
For general $M\geq 3$, when $w>M$ and under other specific parameter settings, we can see from Corollary~\ref{coro:Mchannel_construction} that the lower bound of $K(M,L,w)/L$ approaches
\begin{equation}\label{eq:asym_M_lower}
\frac{M(M-1)}{(2w-M)(w-1)} + \frac{M}{2w-2}
\end{equation}
as $L$ goes to infinity.
However, based on Theorem~\ref{thm:Mchannel_bounds}, our derived upper bound of $K(M,L,w)/L$ approaches 
\begin{equation}\label{eq:asym_M_upper}
\frac{M(M-1)\left(1+(\tau(m)-1)(1-\frac{1}{m})\right)}{(2w-M)(w-1)} + \frac{M}{2w-2}
\end{equation}
as $L$ goes to infinity, where $m=2\frac{w}{M}-1$.
One can see that the asymptotic value in~\eqref{eq:asym_M_lower} reduces to~\eqref{eq:asym_M=1} and \eqref{eq:asym_M=2} when $M=1$ and $2$, respectively.
Motivated by~\eqref{eq:asym_M=1}, we have the following conjecture.
\begin{conjecture}
For all $w\geq M$, we have
\begin{align*}
\limsup_{L\to\infty} \frac{K(M,L,w)}{L} = \frac{M(M-1)}{(2w-M)(w-1)} + \frac{M}{2w-2}.
\end{align*}
\end{conjecture}
This conjecture remains an open problem for future investigation.

\appendices 

\section{A New Proof of Theorem~\ref{thm:UB-any-w}}\label{Appendix:A}

Let $\C$ be any code in $\MCCAC(M,L, w)$. 
Similar to the proof of Theorem~\ref{thm:UB-w>M}, all codewords are non-exceptional and partitioned into the three classes
\begin{align*}
\C_1 &= \{S\in \C: e_S=1\},\\
\C_2 &= \{S\in \C: e_S=w\}, \text{ and} \\
\C_3 &= \{S\in \C: 2\leq e_S \leq w-1\}.\\
\end{align*}

By the same argument in the proof of Theorem~\ref{thm:UB-w>M}, we have
\begin{equation} \label{eq:UB-any-w-1}
(w-1)|\C_1|+\sum_{S\in\C_3}(w-e_S)\leq \frac{M(L-1)}{2},
\end{equation}
and
\begin{equation}\label{eq:UB-any-w-2}
w(w-1)|\C_2| + \sum_{S\in\C_3}(e_S-1)(2w-e_S) \leq M(M-1)L,
\end{equation}
which are analogous to \eqref{eq:UB-w>M-1} and \eqref{eq:UB-w>M-2}, respectively.
Multiplying the inequality in \eqref{eq:UB-any-w-1} by $w$ and combining the result with \eqref{eq:UB-any-w-2}, we obtain
\begin{equation}\label{eq:UB-any-w-3}
w(w-1)(|\C_1|+|\C_2|) + \sum_{S\in\C_3}\left(-e_S^2+(w+1)e_S+w^2-2w\right) \leq \frac{wM(L-1)}{2}+M(M-1)L.
\end{equation}
Define a function $f(x)=-x^2+(w+1)x+w^2-2w$ on the interval $[2,w-1]$.
It is not hard to see $f(x)$ has the minimum value $w^2-2$ at $x=2$ or $x=w-1$, which leads to
\begin{equation}\label{eq:UB-any-w-4}
f(x)\geq w^2-2 \geq w(w-1)
\end{equation}
because of $w\geq 2$.
Hence, it follows from \eqref{eq:UB-any-w-3} that
\begin{equation}\label{eq:UB-any-w-5}
w(w-1)|\C| \leq \frac{wM(L-1)}{2}+M(M-1)L,
\end{equation}
as desired.

Note that in~\eqref{eq:UB-any-w-4}, $f(x)>w(w-1)$ if $w\geq 3$, leading to a strict inequality in~\eqref{eq:UB-any-w-5}.
This concludes that the equality~\eqref{eq:UB-any-w} holds only when all involved codewords are either in $\C_1$ or in $\C_2$, i.e., the packets must either be on the same channel, or else they have to be distributed across different channels.

\bibliography{Ref_MCAC}
\bibliographystyle{IEEEtran}

\end{document}